\documentclass[11pt]{article}
\usepackage{amsmath,amssymb,amsthm}
\usepackage[margin=1.1in]{geometry}
\usepackage{hyperref}

\newtheorem{theorem}{Theorem}[section]
\newtheorem{lemma}[theorem]{Lemma}
\newtheorem{proposition}[theorem]{Proposition}
\newtheorem{corollary}[theorem]{Corollary}
\theoremstyle{definition}

\newtheorem{remark}[theorem]{Remark}

\newcommand{\Z}{\mathbb Z}
\newcommand{\C}{\mathbb C}
\newcommand{\R}{\mathbb R}
\newcommand{\Ht}{\widetilde H}
\newcommand{\Ker}{\operatorname{Ker}}

\newcommand{\sig}{\operatorname{sig}}
\newcommand{\frc}[1]{\{#1\}}
\newcommand{\e}[1]{e\!\left(#1\right)}   %
\newcommand{\PB}{\mathit{PB}}            %

\title{%
A closed signature formula for the Katz--Long--Moody Hermitian form%
}
\author{Haru Negami\\[2pt]
\small Chiba University, 1-33 Yayoi-cho, Inage-ku, Chiba 263-8522, Japan\\
\small \texttt{negamiharu@gmail.com}}
\date{}

\begin{document}
\maketitle

\begin{abstract}
The Katz--Long--Moody construction associates to a representation
of $F_n\rtimes B_n$ and a nonzero parameter $\lambda$ a new
representation of the same group. On the pure braid group
$\PB_{n+1}\cong F_n\rtimes\PB_n$ it corresponds to Haraoka's
multiplicative middle convolution for KZ-type equations. For unitary input and $|\lambda|=1$,
$\lambda\ne1$, the construction equips the quotient representation
with a canonical non-degenerate invariant Hermitian form.
We give a closed formula for its signature in terms of the
eigenangles of the input local monodromies, their ordered product,
and the convolution parameter. The formula accounts for the
kernel of the form before passage to the quotient and for
signature changes at resonant parameters. It determines precisely
when the induced form is definite, answering the definiteness
problem posed in the companion paper. Definiteness implies
unitarizability of the output representation; the converse holds
when that representation is irreducible.

The proof uses elementary linear algebra: a determinant identity,
explicit block-pivot formulas, and an inertia formula for sums
of Cayley transforms of unitary matrices. As applications,
we compare the rank-one construction explicitly with Haraoka's
invariant form for Pochhammer systems, recover the Gauss case
of the Beukers--Heckman interlacing criterion, and determine
the definite parameter intervals for the Hecke and
Temperley--Lieb specializations.
\end{abstract}



\section{Introduction}\label{sec:intro}

The multiplicative middle convolution of a monodromy representation
that is unitary relative to a Hermitian form is again unitary
relative to one, and that form is non-degenerate \cite{Haraoka2020},
\cite{N-KLM}. This does not make the convolved representation
unitary in the usual sense: the canonical form may be indefinite.
The first question is therefore whether this canonical form is
definite. Definiteness implies that the representation is
unitarizable; the converse requires irreducibility, since a
reducible representation may preserve a positive definite form
that is not a multiple of the canonical one
(Remark~\ref{rem:hecke-consequences}(c);
Proposition~\ref{prop:transport}(iii)). The question is then
concrete: the input is a tuple of matrices and a convolution
parameter; how does the signature of the canonical form depend on
the input, and for which parameters is it definite?

Haraoka's multiplicative middle convolution for
Knizhnik--Zamolodchikov (KZ)-type equations \cite{Haraoka2020} has an
algebraic counterpart in the
Katz--Long--Moody (KLM) construction \cite{HN}, \cite{N-KLM}. Here
$F_n\rtimes B_n$ is the semidirect product formed with respect to
the Artin representation $B_n\to\operatorname{Aut}(F_n)$, an action
of $B_n$ on $F_n=\pi_1$ of the punctured disc; we follow the
convention of \cite{N-KLM}. A KZ-type equation is a linear system
$\partial u/\partial z_i=\sum_{j\neq i}A_{ij}(z_i-z_j)^{-1}u$
$(i=0,\dots,n)$ with constant $N\times N$ matrices $A_{ij}=A_{ji}$
subject to the integrability conditions \cite[Def.~13]{N-KLM}. The KZ
equation itself was derived in conformal field theory as a differential equation satisfied by $n$-point
correlation functions \cite{KZ}. To a representation
of $F_n\rtimes B_n$ on $V$ and a nonzero parameter $\lambda$ the KLM
construction associates a new representation of the same group,
obtained by passing from $V^{\oplus n}$ to a quotient determined by
the input matrices and $\lambda$; on the pure braid group
$\PB_{n+1}\cong F_n\rtimes\PB_n$ it corresponds to the
multiplicative middle convolution for KZ-type equations, and its
rank-one case contains
the hypergeometric and Pochhammer systems and the Burau
representation. For unitary input and $|\lambda|=1$,
$\lambda\neq1$, the quotient representation carries a canonical
non-degenerate invariant Hermitian form \cite{N-KLM}, and the
companion paper gave a recursive procedure for computing its
signature. In this paper we
replace this recursive computation by a closed formula, valid also
when the input matrices have eigenvalue $1$ and when the convolution
parameter lies on a resonance wall.

The formula uses only the eigenangles of the input matrices
$g_1,\dots,g_n$, those of their ordered product $P_n=g_1\cdots g_n$,
and the parameter $\lambda$; nothing has to be diagonalized along
the way. In the
local-system interpretation, the spectrum of $P_n^{-1}$ records the
monodromy at infinity; thus the formula uses the local spectral data
at all punctures. In particular, for fixed $\lambda$, the signature is
constant among unitary tuples with the same local conjugacy classes,
even when the tuples themselves vary in a positive-dimensional family.

The signature computation is a statement about arbitrary unitary
tuples and does not require braid compatibility. We therefore state
the main theorem first at the tuple level, and then explain its
consequences for the KLM representations arising from unitary
braid-compatible input.

To state the result, let $g_1,\dots,g_n\in U(N)$, let
$\lambda=\e{l}:=\exp(2\pi il)$ with $l\in(0,1)$, and let
$\Ht=\Ht(\lambda)$ be the
invariant Hermitian form of the construction (recalled in
Section~\ref{sec:prelim}). Write $P_n:=g_1\cdots g_n$ for the
ordered product of the input matrices --- the letter $P$ with a
numerical subscript is reserved for these partial products
throughout, the pure braid group being written $\PB_n$ --- and let
$a_j^{(k)}\in[0,1)$ $(1\le j\le n,\ 1\le k\le N)$ and
$\beta_k(l)\in[0,1)$ $(1\le k\le N)$ be the eigenangles of $g_j$ and
of $\lambda P_n$, listed with algebraic multiplicity. Set
\[
K:=\bigoplus_j\Ker(g_j-1),\quad \kappa:=\dim K,\quad R:=nN-\kappa,\quad
r(l):=\dim\Ker(\lambda P_n-1),
\]
\[
m(l):=Nl+\sum_{j,k}a_j^{(k)}-\sum_k\beta_k(l)\ \in\Z ,
\]
and let $L=L(\lambda)\subseteq V^{\oplus n}$ be the
Dettweiler--Reiter subspace, cut out by $v_k=g_{k+1}v_{k+1}$
$(1\le k\le n-1)$ and $v_n=\lambda P_nv_n$; the map $v\mapsto v_n$
identifies it with $\Ker(\lambda P_n-1)$, so $\dim L=r(l)$.

\begin{theorem}[Main Theorem, tuple level; = Cor.~\ref{cor:quotsig}]
\label{thm:main-intro}
Let $g_1,\dots,g_n\in U(N)$ be arbitrary. For every $l\in(0,1)$ ---
that is, on the entire admissible parameter space
$\lambda\in S^{1}\setminus\{1\}$ --- the Hermitian form $\Ht(\lambda)$
on $V^{\oplus n}$ has kernel $K\oplus L$, and the form it induces
on $V^{\oplus n}/(K\oplus L)$ is non-degenerate, of signature
\[
\bigl(R-m(l),\;m(l)-r(l)\bigr),
\]
with $r(l)=0$ off the finitely many walls
$\lambda^{-1}\in\operatorname{spec}P_n$. In particular: the induced
form on the quotient is definite if and only if $m\in\{r,R\}$
(positive definite iff $m=r$, negative definite iff $m=R$); the
ambient form $\Ht$ is then semidefinite, with kernel $K\oplus L$ of
dimension $\kappa+r$, and is itself definite if and only if moreover
$\kappa+r=0$.
\end{theorem}

\begin{corollary}[braid-compatible input]\label{cor:main-braid}
If the tuple comes with \emph{unitary} braid data --- a
representation $\rho\colon F_n\rtimes B_n\to U(N)$ for one and the
same inner product, the input of the KLM construction --- then $K\oplus L$
is the radical of the invariant form of the KLM representation, the
KLM quotient $V^{\oplus n}/(K\oplus L)$ is an
$F_n\rtimes B_n$-representation, and
Theorem~\ref{thm:main-intro} computes the signature of its canonical
invariant form. Moreover the representation of $B_n$ on
$V^{\oplus n}/K$ does not depend on $\lambda$ off the walls
(Proposition~\ref{prop:transport}): $\lambda$ parametrizes the pencil
of invariant forms, not the representation.
\end{corollary}

Unitarity of the braid part is not a formality: for $n=2$, $N=1$,
$g_1=g_2=t$ and $\rho(\sigma_1)=2$ the data are compatible, but the
output braid generator has eigenvalues $2$ and $-2t$ and preserves no
definite form. Not every unitary tuple carries braid data either: the
non-rigid example of
Section~\ref{subsec:nonrigid} has seeds with pairwise distinct spectra
and is not an $F_n\rtimes B_n$ input; the theorem still applies to it,
the corollary does not.

For a unitary representation $\rho\colon F_n\rtimes B_n\to U(N)$,
the form in the theorem is invariant under the induced action of
$F_n\rtimes B_n$, in particular under the braid action
\cite[Thm.~8]{N-KLM}. Its definiteness therefore implies that the
KLM quotient is unitarizable, as a representation of
$F_n\rtimes B_n$ and a fortiori of $B_n$. The theorem determines the
definite locus of this particular canonical form; it does not, for
reducible quotients, classify the existence of all possible positive
definite invariant forms --- a direct sum $H_1\oplus(-H_2)$ is
indefinite while the representation preserves $H_1\oplus H_2$, and
Remark~\ref{rem:hecke-consequences}(c) exhibits an instance. For an
irreducible quotient the converse holds as well
(Proposition~\ref{prop:transport}(iii)): the invariant Hermitian
form is then unique up to a real scalar, so unitarizability as a
representation of $F_n\rtimes B_n$ is equivalent to definiteness of
the canonical form when the quotient is irreducible under
$F_n\rtimes B_n$ --- the irreducibility supplied by
\cite[\S4.3]{N-KLM} when the input is irreducible as an
$F_n$-representation, provided that the KLM quotient is nonzero ---
and likewise for $B_n$ when it is
irreducible under $B_n$; the weaker $F_n\rtimes B_n$-irreducibility
does not suffice for the $B_n$-statement, since a $B_n$-invariant
form need not be $F_n$-invariant. Throughout, ``definite locus''
refers to the canonical form.

A further distinction concerns the role of $\lambda$. For fixed
input, the braid operators on $V^{\oplus n}/K$ are independent of
$\lambda$. Away from the resonance walls, $\lambda$ therefore varies
the canonical invariant form on a fixed braid representation. At a
wall, the additional subspace $L(\lambda)$ is removed. The definite
windows described below should be understood in this sense.

The formula is uniform: the kernel locus $K\neq0$ and the resonance
walls $\lambda\in E$ require no case distinction, no perturbation off
the degenerate set, and no positivity or rigidity hypotheses on the
input beyond unitarity. The single convention that eigenangles are
taken in $[0,1)$ decides all boundary values, and it is the same
convention in which the Beukers--Heckman interlacing criterion is
stated \cite{BH} --- a coincidence that Section~\ref{sec:BH} explains.

\medskip\noindent
\textbf{Relation to \cite{N-KLM}.} The companion paper \cite{N-KLM} established
the invariance theory of $\Ht$ and gave a recursive \emph{algorithm}
computing its signature, requiring the sequential diagonalization of
pivot matrices. The present paper replaces the algorithm by a closed
formula: the inertia is read off two spectra alone --- the eigenangles
of the seeds and of $\lambda P_n$ --- with no intermediate
diagonalization and no dependence on the relative position of
eigenvectors, which enters only through
$\operatorname{spec}(\lambda P_n)$; it applies uniformly to the
kernel and resonance loci, without the partial-product invertibility
assumption used in the companion paper; and it classifies
definiteness of the canonical form completely, thereby solving
\cite[Problem~18]{N-KLM} --- the definiteness question that the
recursive algorithm of \cite{N-KLM} leaves open --- for unitary
input. The radical $\Ker\Ht=K\oplus L$, the non-degeneracy of the
induced form on the quotient, and the vanishing of $\Ht$ on $K$ are
results of \cite[Thm.~9 and Cor.~13]{N-KLM}, used here at tuple level
(and in the positive definite case of the seed form, which is the case
of \cite[Problem~18]{N-KLM}); the degenerate-locus section combines
this permanent-kernel property with two crossing forms to determine
the signature on these loci.

The proof is purely linear-algebraic, and its engine deserves separate
billing. For a unitary $u$ without eigenvalue $1$ let
$C(u):=i(1+u)(1-u)^{-1}$ be its (Hermitian) Cayley transform and
$\Theta(u)\in(0,N)$ the sum of its eigenangles in $(0,1)$.

\begin{lemma}[Cayley inertia lemma; = Lem.~\ref{lem:cayley}]
\label{lem:cayley-intro}
If $1\notin\operatorname{spec}a\cup\operatorname{spec}b\cup
\operatorname{spec}(ab)$, then
\[
n_{+}\bigl(C(a)+C(b)\bigr)=\Theta(a)+\Theta(b)-\Theta(ab).
\]
\end{lemma}

Identities relating inertia counts to eigenangle data are classical
in the theory of signature non-additivity, cf.\ \cite{Wall}; we make
no claim of novelty for the lemma, and give a short self-contained
proof by a monotone-transversality argument carried out entirely in
Cayley charts.
The signature theorem then follows by a
pivot (Schur-complement) recursion whose $s$-th step contributes
exactly $m_s-m_{s-1}$ negative eigenvalues --- a telescoping of the
cocycle --- and the degenerate loci are handled by two crossing forms
of opposite definite signs,
\[
\Gamma_K=4\pi\sin(\pi l)\,I\succ0,
\qquad
\Gamma_L=-4\pi\sin(\pi l)\,\lVert v_n\rVert^2\prec0,
\]
combined with Haynsworth inertia additivity; the permanent-kernel
property $\Ht(\lambda)K=0$, identically in $\lambda$ (a one-line
consequence of the block structure, already noted in
\cite[Thm.~9]{N-KLM}), shows that the two degeneracies never
interact. No Hodge theory, no
analytic perturbation theory, and no eigenvector branches enter.

Two consistency anchors tie the theorem to existing computations. On
the classical side, suitable rank-one specializations of the KLM
representation recover the monodromy of Pochhammer (in particular
hypergeometric) systems. Two separate statements are proved there.
First, an explicit matrix identity: in rank one the form $\Ht$ and
Haraoka's invariant form \cite{Haraoka1994} are related entrywise by
a congruence up to sign, which holds for all parameters covered by
Proposition~\ref{prop:haraoka-entrywise}, resonant ones included.
Second,
under the non-resonance conditions of \cite{Haraoka1994},
Theorem~\ref{thm:main-intro} recovers Haraoka's definiteness
conditions and refines them to a signature. For the Beukers--Heckman
theory \cite{BH} the situation is different in the two cases treated.
In the Gauss case ($n=2$, $N=1$) we prove that the signature of
Theorem~\ref{thm:main-intro} is the Beukers--Heckman signature, up
to the scalar ambiguity of the invariant form. In
general rank we give a \emph{reduction}: for disjoint output data
satisfying the hypotheses of the Beukers--Heckman signature theorem,
that theorem follows from Theorem~\ref{thm:main-intro} once the local
monodromy data of the middle convolution are identified with the
hypergeometric presentation; this identification is not carried out
here, so no general-rank correspondence is claimed as proved. The
integer $m$ appears as the interlacing-defect count. Signatures of
this kind can also be approached through the Hodge theory of the
middle convolution, developed by Dettweiler--Sabbah \cite{DS}; we do
not pursue the comparison here.

On the quantum side, the companion paper \cite{N-Hecke} classified when the
KLM construction with scalar braid part produces Iwahori--Hecke and
Temperley--Lieb modules, and posed the definiteness question for
their invariant forms as a specialization of \cite[Problem~18]{N-KLM}.
Theorem~\ref{thm:main-intro} answers it completely.

\begin{corollary}[= Cor.~\ref{cor:hecke}]
\label{cor:hecke-intro}
For the scalar family with $\operatorname{spec}(g)\subseteq\{1,\e a\}$,
$d:=\dim\Ker(g-\e a)$, the invariant form on the Hecke module has
signature $\bigl(d(n-j),\,dj\bigr)$ with $j=\lfloor l+na\rfloor$ off
resonance, and $\bigl(d(n-J),\,d(J-1)\bigr)$ at the resonance
$l+na=J$; a definite window in $\lambda$ exists if and only if
$a\in(0,\tfrac1n)\cup(\tfrac{n-1}n,1)$ --- the classical
Burau--Squier interval --- uniformly in the multiplicities, while the
double-resonant $n=2$ family and the $\mu=-1$ family are indefinite
for every admissible $\lambda$.
\end{corollary}

The definite windows yield unitary Temperley--Lieb representations.
Of the two always-indefinite families, the $\mu=-1$ family is
non-unitarizable outright, its Jordan block being the obstruction,
whereas the double-resonant $n=2$ family is unitarizable after a
channelwise sign change (Remark~\ref{rem:hecke-consequences}(c)): the
indefiniteness of the canonical form is not, by itself, a failure of
unitarity. This feeds directly into the study of
unitary braid representations and their images \cite{Jones}, \cite{Wenzl},
\cite{FLW}, \cite{LR}, and into the recent interest in braid representations
that are unitary with respect to \emph{indefinite} forms \cite{GLPS};
sesquilinear structure is also the entry point for discreteness
constructions as in \cite{Scherich}.

\medskip\noindent
\textbf{Organization.} Section~\ref{sec:prelim} recalls the KLM
construction and the form $\Ht$ and fixes the $[0,1)$-conventions and
the integer $m(l)$. Section~\ref{sec:cayley} proves the Cayley
inertia lemma. Section~\ref{sec:pivot} establishes the
determinant identity and the pivot recursion under the PI condition;
Section~\ref{sec:open} proves the signature theorem on the open
locus. Section~\ref{sec:degenerate} treats the degenerate loci and
proves Theorem~\ref{thm:main-intro}. Section~\ref{sec:BH} is the
Beukers--Heckman/Haraoka dictionary, the Hecke/Temperley--Lieb
specialization (\S\ref{sec:hecke}) and a non-rigid example.
Section~\ref{sec:discussion} discusses indefinite seed forms and open
problems.

\section{Preliminaries: the KLM construction and its Hermitian form}
\label{sec:prelim}
The following notation is in force throughout the paper.
Throughout, $j\in\{1,\dots,n\}$ indexes the seeds $g_j$ and
$k\in\{1,\dots,N\}$ indexes the eigenvalues of a fixed unitary.
\begin{itemize}
\item[(a)] \emph{Scalars.} $e(x):=\exp(2\pi ix)$; $\frc{x}\in[0,1)$
denotes the fractional part. Eigenvalues of unitaries are always
listed with algebraic multiplicity, so that sums over eigenangles ---
such as $\Theta(u)$ and $m(l)$ below --- are multiset sums.
\item[(b)] \emph{Tuple data.} $g_1,\dots,g_n\in U(N)$ with eigenangles
$a_j^{(k)}\in(0,1)$ ($j$: which seed, $k$: which eigenvalue);
$\lambda=e(l)$, $l\in(0,1)$, $\lambda^{1/2}=e(l/2)$;
$P_s:=g_1\cdots g_s$ with $P_0:=I$, eigenangles
$\gamma^{(s)}_k\in[0,1)$ of $P_s$ and
$\beta^{(s)}_k=\frc{l+\gamma^{(s)}_k}$ of $\lambda P_s$. The point
$\lambda=1$ is excluded from the ambient space throughout
($l\in(0,1)$).
\item[(c)] \emph{Winding integers.}
$m_s:=Nl+\sum_{j\le s,k}a_j^{(k)}-\sum_k\beta^{(s)}_k\in\Z$ (so
$m_0=0$, since $\beta^{(0)}_k=l$ for all $k$), and $m:=m_n$.
\item[(d)] \emph{Inertia counts.} For a Hermitian matrix $A$,
$n_{+}(A)$ denotes the number of positive and $q(A)$ the number of
negative eigenvalues, both with multiplicity, and
$\widetilde q:=q(\Ht)$.
\item[(e)] \emph{Loewner order.} $A\succ0$ ($A\succeq0$) means
positive (semi)definite, $A\prec0$, $A\preceq0$ the negative
versions, and $A\preceq B$ means $B-A\succeq0$.
\item[(f)] \emph{The open locus and the walls.} The proof first runs
on an open dense locus on which every factorization below exists; the
degenerate loci are added in Section~\ref{sec:degenerate} below.
Define the wall set $E$ through the equivalences
\[
\lambda\in E
\iff \lambda^{-1}\in\operatorname{spec}P_n
\iff 1\in\operatorname{spec}(\lambda P_n)
\iff \dim\Ker(\lambda P_n-1)\ge1 ,
\]
and let
$\mathcal P^{\circ}:=\{K=0,\ \lambda\notin E\}
\subset U(N)^{n}\times(S^{1}\setminus\{1\})$ be the open
non-degenerate locus of tuples-with-parameter
$(g_1,\dots,g_n;\lambda)$: no $g_j$ has eigenvalue $1$ (all
$a_j^{(k)}\in(0,1)$, i.e.\ $K=0$ --- equivalently, every $g_j-1$,
and hence the block-diagonal $G^{\mathrm{diag}}-1$ below, is
invertible), and $\lambda\notin E$. On $\mathcal P^{\circ}$ the form
$\Ht$ is non-degenerate and $\widetilde q$, $m$ are locally
constant.
\end{itemize}
Two cautions, kept apart from the definitions: $\gamma^{(s)}_k=0$
can occur even though no $g_j$ has eigenvalue $1$ --- the
determinant identity and the pivot formulas below are insensitive to
this, and the PI condition concerns $\beta^{(s)}_k\neq0$ only; and
the PI condition is \emph{not} part of the definition of
$\mathcal P^{\circ}$ and may fail there.

For a unitary $u$ with $1\notin\operatorname{spec}u$ write
$\Theta(u):=\sum_k\theta_k\in(0,N)$ for the sum of its eigenangles
$\theta_k\in(0,1)$, and
\begin{equation}\label{eq:cayley-def}
C(u):=i(1+u)(1-u)^{-1},
\qquad
C(e(\theta))
=i\,\frac{2e^{i\pi\theta}\cos\pi\theta}{-2i\,e^{i\pi\theta}\sin\pi\theta}
=-\cot\pi\theta ,
\end{equation}
the Cayley transform: a Hermitian matrix whose spectral map
$\theta\mapsto-\cot\pi\theta$ is an increasing bijection
$(0,1)\to(-\infty,\infty)$.

\paragraph{Degenerate loci.} The proof of the signature theorem
first runs on $\mathcal P^{\circ}$; Section~\ref{sec:degenerate}
then allows the seeds to have eigenvalue $1$ and $\lambda$ to lie on a
wall. The following conventions extend the ones above, agree with
them on $\mathcal P^{\circ}$, and fix once and for all the values on
the degenerate loci. Eigenangles are taken in $[0,1)$:
$a_j^{(k)}\in[0,1)$ for $g_j$, and $\beta_k=\beta_k(l)\in[0,1)$ for
$\lambda P_n$. Set $W_j:=\Ker(g_j-1)$, $\kappa_j:=\dim W_j$,
\[
K:=W_1\oplus\dots\oplus W_n\subseteq V^{\oplus n},\qquad
\kappa:=\dim K=\textstyle\sum_j\kappa_j,\qquad
R:=nN-\kappa,
\]
$r=r(l):=\dim\Ker(\lambda P_n-1)=\#\{k;\ \beta_k(l)=0\}$ (the
letter $r$ is for the \emph{resonance} multiplicity --- $l$ being
reserved for the angle of $\lambda$ --- and the KLM quotient below
has dimension $R-r$), and, uniformly
in $l\in(0,1)$,
\[
m=m(l):=Nl+\sum_{j,k}a_j^{(k)}-\sum_k\beta_k(l)\ \in\Z ;
\]
integrality holds on the degenerate loci as well, since
$e(m)=\lambda^{N}\det P_n\cdot\det(\lambda P_n)^{-1}=1$. On
$\mathcal P^{\circ}$ this is the previous $m$; on the walls the
$[0,1)$-convention selects a specific value, which
Proposition~\ref{prop:wall} identifies as the right-sided limit. The
subspace $L\subseteq V^{\oplus n}$ is the Dettweiler--Reiter subspace,
characterized by $v_k=g_{k+1}v_{k+1}$ $(1\le k\le n-1)$,
$v_n=\lambda P_nv_n$ (\cite[Eq.~(9)]{N-KLM}); $v\mapsto v_n$ identifies
$L$ with $\Ker(\lambda P_n-1)$, so $\dim L=r$. For $\lambda\neq1$ one
has
\begin{equation}\label{eq:radical}
\Ker\Ht=K\oplus L .
\end{equation}
The inclusion $\supseteq$ and the directness are
Lemmas~\ref{lem:permker} and \ref{lem:radical} (the directness is
also shown in the proof of \cite[Thm.~16(iii)]{N-KLM}); the reverse
inclusion is cited from \cite[Thm.~9]{N-KLM}, whose tuple-level
nature is noted there before Corollary~13 (which applies it to the
truncated tuples), in Proposition~\ref{prop:radical} below.
Nothing in Sections~\ref{sec:cayley}--\ref{sec:degenerate} uses the
$F_n\rtimes B_n$ structure: all statements concern arbitrary unitary
tuples.

\paragraph{The setting of \cite{N-KLM}, and the form $\Ht$.} In
\cite{N-KLM} the input is a representation of $F_n\rtimes B_n$ that
is unitary relative to a non-degenerate Hermitian matrix $H$, which
may be indefinite \cite[Def.~7]{N-KLM}; the invariance of $\Ht$ and
the identification of its kernel \cite[Thms.~8, 9]{N-KLM} hold in
that generality, with
$\Ht_{jk}=\lambda^{jk}H(g_j^{-1}-\lambda_{jk}I)(g_k-1)$, where
$\lambda^{jk}=\lambda^{-1/2}$ for $j\le k$ and $\lambda^{1/2}$ for
$j>k$, and $\lambda_{jk}=\lambda$ for $j=k$ and $1$ otherwise. The
present paper restricts this setting to positive definite $H$ ---
which is exactly the input condition of the definiteness question
\cite[Problem~18]{N-KLM}, so that the restriction and the claim to
solve that problem are compatible. We assume $H>0$ throughout and
choose an $H$-orthonormal basis of $V$, so that $H=I$ and the $g_j$
are unitary in the usual sense; in that basis the blocks of $\Ht$ are
\begin{equation}\label{eq:blocks}
\Ht_{jk}=
\begin{cases}
\lambda^{-1/2}\,(g_j^{-1}-1)(g_k-1), & j<k,\\[1pt]
\lambda^{-1/2}\,(g_j^{-1}-\lambda)(g_j-1), & j=k,\\[1pt]
\lambda^{1/2}\,(g_j^{-1}-1)(g_k-1), & j>k.
\end{cases}
\end{equation}

\section{The Cayley inertia lemma}\label{sec:cayley}

Recall the Cayley transform \eqref{eq:cayley-def} and
$\Theta(u)=\sum_k\theta_k\in(0,N)$ from the head of the section; the
letter $\sigma$ is reserved for the braid generators $\sigma_i$, and
$n_{+}$ counts positive eigenvalues with multiplicity. Three
classical facts about \eqref{eq:cayley-def} are used below. $C(u)$ is
Hermitian: using $u^{\dagger}=u^{-1}$ and the commutativity of
functions of $u$,
$C(u)^{\dagger}=-i(1+u^{-1})(1-u^{-1})^{-1}=-i(u+1)(u-1)^{-1}=C(u)$.
Its eigenvalues are $-\cot(\pi\theta_k)$: $u$ is unitarily
diagonalizable and $C$ is a rational function of $u$, so $C(u)$ is
diagonal in the same eigenbasis with the scalar values of
\eqref{eq:cayley-def}, increasing in the eigenangle. And $C$ is a
bijection from $\{u\in U(N)\mid 1\notin\operatorname{spec}u\}$ onto
the Hermitian matrices, with inverse $u=(C+i)^{-1}(C-i)$; the proof
below uses this inverse to pull straight-line deformations on the
Hermitian side back to the unitary side.

\begin{lemma}\label{lem:cayley}
Let $a,b\in U(N)$ with
$1\notin\operatorname{spec}a\cup\operatorname{spec}b
\cup\operatorname{spec}(ab)$. Then the number of positive eigenvalues of
$C(a)+C(b)$ is
\[
n_{+}\bigl(C(a)+C(b)\bigr)=\Theta(a)+\Theta(b)-\Theta(ab).
\]
\end{lemma}

\begin{remark}[the resonant case $1\in\operatorname{spec}(ab)$]
\label{rem:cayley-resonant}
The hypotheses on $a$ and $b$ are essential ($C$ has a pole at $1$),
but the one on $ab$ is not: $C(a)+C(b)$ is defined whenever
$1\notin\operatorname{spec}a\cup\operatorname{spec}b$. Let
$\Theta_0(ab)$ be the sum of the eigenangles of $ab$ taken in
$[0,1)$, $r:=\dim\Ker(ab-1)$ and
$M:=\Theta(a)+\Theta(b)-\Theta_0(ab)$. Then
\[
\operatorname{In}\bigl(C(a)+C(b)\bigr)=(M-r,\;N-M,\;r).
\]
Indeed, run the deformation of the proof below, $M_t=C(a)+C(b)-tI$,
$b_t=C^{-1}(C(b)-tI)$, for small $t>0$: $\Ker M_0$ has dimension $r$
(it corresponds to $\Ker(ab-1)$ through the chart, see the proof),
its $r$ eigenvalues move to $-t<0$, and the remaining eigenvalues keep
their signs, so $n_{+}(M_0)=n_{+}(M_t)$. On the unitary side
$1\notin\operatorname{spec}(ab_t)$ for $t>0$, the $r$ eigenangles of
$ab_t$ that vanish at $t=0$ leave $0$ through $1^{-}$ (this is the
monotonicity of the chart used in the proof), and
$\Theta(b_t)\to\Theta(b)$; the lemma at $t>0$ then gives
$n_{+}(M_0)=\Theta(a)+\Theta(b)-\Theta_0(ab)-r=M-r$. The extreme case
$b=a^{-1}$ ($r=N$, $M=N$) gives the zero matrix.
\end{remark}

In the pivot argument the lemma is invoked off the walls, where the
hypothesis on $ab$ holds: with $a=g_s$, $b=\lambda P_{s-1}$, so that
$ab=\lambda g_sP_{s-1}$ is similar to $\lambda P_s$ (the spectra
coincide), the three conditions hold on the locus $K=0$, PI, on which
the pivot argument runs; this is contained in, and for $n\ge2$
strictly smaller than, the open locus $\mathcal P^{\circ}$ of
Theorem~\ref{thm:rankN}, whose statement is then extended by
continuity (Section~\ref{sec:degenerate}). Tuples with eigenvalue $1$,
which dominate the applications (the Burau seeds foremost), are
reached by the one-sided limits of Section~\ref{sec:degenerate}, not
by extending the Cayley transform; Remark~\ref{rem:cayley-resonant}
serves only as a consistency check on the walls.

For $N=1$ this is
$-\cot\pi x-\cot\pi y=-\sin\pi(x+y)/(\sin\pi x\sin\pi y)>0
\iff\lfloor x+y\rfloor=1$.

\begin{proof}
Set $B_t:=C(b)-tI$ and $b_t:=C^{-1}(B_t)=(B_t-i)(B_t+i)^{-1}$ for
$t\in[0,\infty)$: since $B_t$ is Hermitian, $b_t$ is unitary with
$1\notin\operatorname{spec}b_t$ for every $t$, and $b_0=b$. Put
$U_t:=ab_t$ and
\[
M_t:=C(a)+C(b_t)=C(a)+C(b)-tI
=-2i\,(a-1)^{-1}(U_t-1)(b_t-1)^{-1},
\]
the last equality by the sum identity used in
Theorem~\ref{thm:pivot}. In this factorization the outer factor
$(a-1)^{-1}$ and the inner factor $(b_t-1)^{-1}$ are invertible for
every $t$, so the kernel is governed by the middle factor alone:
$M_t$ is not invertible iff $1\in\operatorname{spec}U_t$, and
$\ker M_t=(b_t-1)\ker(U_t-1)$, of the same dimension ($a$ enters
through $U_t=ab_t$); and since the
eigenvalues of $M_t$ are $\mu_j-t$, $\mu_j$ the eigenvalues of
$C(a)+C(b)$, the crossing set
$Z:=\{t\ge0\mid 1\in\operatorname{spec}U_t\}=\{\mu_j\}\cap[0,\infty)$
is finite, with $0\notin Z$ by hypothesis. Off $Z$ define
\[
F(t):=n_{+}(M_t),\qquad
G(t):=\Theta(a)+\Theta(b_t)-\Theta(ab_t).
\]

\emph{Both vanish for $t\gg0$.} For $t>\max_j\mu_j$ we have
$M_t\prec0$, so $F=0$. Since $C(b_t)=C(b)-tI$ by construction, its eigenvalues are
$\beta_k-t$ ($\beta_k\in\operatorname{spec}C(b)$, eigenbasis fixed);
comparing with \eqref{eq:cayley-def}, the eigenangles
$\theta_k(t)\in(0,1)$ of $b_t$ solve $-\cot\pi\theta_k=\beta_k-t$,
hence decrease to $0^{+}$ without leaving $(0,1)$; thus $\Theta(b_t)$ is continuous on all of
$[0,\infty)$ and tends to $0$, while $ab_t\to a$ and
$\Theta(ab_t)\to\Theta(a)$. Moreover
$e\bigl(G(t)\bigr)=\det a\cdot\det b_t\cdot\det(ab_t)^{-1}=1$, so $G$
is integer-valued; being continuous off $Z$, it is locally constant
there, and $G=0$ for $t\gg0$.

\emph{Jumps of $F$.} As $t$ increases through $\mu\in Z$, $F$ drops by
exactly $m(\mu):=\dim\ker M_{\mu}=\dim\ker(U_{\mu}-1)$.

\emph{Uniform monotone transversality for $U_t$.} With
$\Omega_t:=-iU_t^{-1}\dot U_t$. Since $\dot B_t=-I$, the product rule
on $b_t=(B_t-i)(B_t+i)^{-1}$ (functions of $B_t$ commute) gives
\begin{gather*}
\dot b_t=(B_t+i)^{-2}\bigl[-(B_t+i)+(B_t-i)\bigr]=-2i\,(B_t+i)^{-2},\\
\Omega_t=-i\,b_t^{-1}\dot b_t=-2\,(B_t^{2}+I)^{-1}
\preceq-\frac{2}{1+\|B_t\|^{2}}\,I\prec0 ,
\end{gather*}
and the sign of $\Omega_t$ controls every Cayley chart of the path: for
$u_0\in S^{1}$, on any $t$-interval with
$u_0\notin\operatorname{spec}U_t$ set $W_t:=u_0^{-1}U_t$ and
$H_t:=C(W_t)$; then, differentiating and substituting
$\dot W_t=iW_t\Omega_t$,
\[
\dot H_t=-2\,(1-W_t)^{-1}W_t\Omega_t(1-W_t)^{-1},
\qquad
x^{\dagger}\dot H_t\,x=2\,z^{\dagger}\Omega_t\,z
\quad\bigl(x=(1-W_t)z\bigr),
\]
using $(1-W_t^{\dagger})(1-W_t)^{-1}W_t=-I$. Thus $\dot H_t\prec0$
wherever the chart is defined. Now fix $\mu\in Z$ and choose
$u_0\in S^{1}\setminus\bigl(\operatorname{spec}U_{\mu}\cup\{1\}\bigr)$.
The path $t\mapsto U_t$ is smooth ($b_t$ is a rational function of
$B_t$, with $B_t+i$ always invertible) and its eigenvalues move
continuously, so $u_0\notin\operatorname{spec}U_t$ persists on a
compact neighbourhood of $\mu$: there the chart is defined --- note
that it excludes $u_0$, not $1$, so the crossing of $1$ at $t=\mu$ is
watched from \emph{inside} the chart --- and the
bound is quantitative: since $\lVert(1-W_t)z\rVert\le2\lVert z\rVert$,
the chart identity gives
$x^{\dagger}\dot H_tx\le-(1+\lVert B_t\rVert^{2})^{-1}\lVert
x\rVert^{2}$, i.e.\ $\dot H_t\preceq-c_tI$ with
$c_t=(1+\lVert B_t\rVert^{2})^{-1}$ continuous and positive; taking
$c:=\min c_t>0$ over the neighbourhood, Weyl's inequality makes every
min--max eigenvalue of $H_t$ strictly decreasing there. Since the
spectral map of the chart is an increasing bijection from the cut
circle $S^{1}\setminus\{u_0\}$ to $\R$, exactly $m(\mu)$ eigenvalues
of $H_t$ pass through the chart value of $1$, at $t=\mu$ and strictly
downward; equivalently, exactly $m(\mu)$ eigenangles of $U_t$ cross
$0$, their representatives in $(0,1)$ jumping from $0^{+}$ to $1^{-}$.
Hence $\Theta(ab_t)$ jumps by $+m(\mu)$ and $G$ drops by exactly
$m(\mu)$.

$F$ and $G$ are therefore locally constant off the finite set $Z$,
drop by the same amount at every point of $Z$, and vanish for
$t\gg0$; hence $F(0)=G(0)$, which is the assertion.

\end{proof}

\section{The determinant identity and the pivot recursion}
\label{sec:pivot}
\subsection{Block minors}

Call a block matrix $B\in M_{sN}(\C)$ \emph{tri-constant} (our
terminology) if its strictly upper blocks are all equal to one
constant multiple of the identity, its strictly lower blocks to
another, and its diagonal blocks are arbitrary:
\[
B=\begin{pmatrix}
D_1 & \alpha I & \alpha I & \cdots & \alpha I\\
\beta I & D_2 & \alpha I & \cdots & \alpha I\\
\beta I & \beta I & D_3 & \ddots & \vdots\\
\vdots & \vdots & \ddots & \ddots & \alpha I\\
\beta I & \beta I & \cdots & \beta I & D_s
\end{pmatrix}.
\]

\begin{lemma}[tri-constant block determinant]\label{lem:blocktri}
Let $\alpha\neq\beta\in\C$, $D_1,\dots,D_s\in M_N(\C)$ with $D_j-\alpha I$
invertible, and let $B\in M_{sN}(\C)$ be tri-constant with blocks
$\alpha I$ ($j<k$), $\beta I$ ($j>k$), $D_j$ ($j=k$). Then, with
$T_s:=\prod_{j=s}^{1}(D_j-\alpha)^{-1}(D_j-\beta)$,
\[
\det B=\Bigl(\prod_{j}\det(D_j-\alpha)\Bigr)
\frac{\det(\beta I-\alpha T_s)}{(\beta-\alpha)^{N}} .
\]
\end{lemma}

\begin{proof}
$B=B_0+\alpha EE^{T}$, $E=(I,\dots,I)^{T}$, $B_0$ block lower triangular
with diagonal $D_j-\alpha$, strictly lower blocks $(\beta-\alpha)I$.
Solving $B_0Y=E$ by blockwise forward substitution and setting
$T_j:=I-(\beta-\alpha)(Y_1+\dots+Y_j)$, $T_0=I$: row $j$ reads
$(D_j-\alpha)Y_j=T_{j-1}$, whence
$T_j=\bigl[I-(\beta-\alpha)(D_j-\alpha)^{-1}\bigr]T_{j-1}
=(D_j-\alpha)^{-1}(D_j-\beta)\,T_{j-1}$; and, rearranging the
definition of $T_s$,
\[
E^{T}B_0^{-1}E=\sum_jY_j=(\beta-\alpha)^{-1}(I-T_s).
\]
By Sylvester's determinant identity
$\det(I_m+XY)=\det(I_n+YX)$ (see, e.g., \cite{HJ}; or take the two Schur
complements of $\bigl(\begin{smallmatrix}I&-X\\Y&I
\end{smallmatrix}\bigr)$), applied with $X=\alpha B_0^{-1}E$,
$Y=E^{T}$, the $sN\times sN$ determinant collapses to $N\times N$:
\[
\det B=\det\bigl(B_0+(\alpha E)E^{T}\bigr)
=\det B_0\,\det\bigl(I_{sN}+\alpha B_0^{-1}EE^{T}\bigr)
=\det B_0\,\det\bigl(I_N+\alpha E^{T}B_0^{-1}E\bigr),
\]
and here
\[
I_N+\alpha E^{T}B_0^{-1}E
=I_N+\tfrac{\alpha}{\beta-\alpha}(I-T_s)
=\frac{\beta I-\alpha T_s}{\beta-\alpha}:
\]
the determinant of this $N\times N$ matrix contributes
the factor $(\beta-\alpha)^{-N}$.
\end{proof}

Let $\Ht^{[s]}$ denote the leading principal $sN\times sN$ block of
$\Ht$; since the block $(j,k)$ of $\Ht$ depends only on $g_j,g_k$ and
$\lambda$, this is also the form of the truncated tuple
$(g_1,\dots,g_s)$, and $\Ht^{[n]}=\Ht$. On the locus
$a_j^{(k)}\in(0,1)$ for all $j,k$, set
$G^{\mathrm{diag}}:=\operatorname{diag}(g_1,\dots,g_s)$
(the letter $G$ alone is reserved for the middle-convolution output,
as in \cite{N-KLM}), so that $G^{\mathrm{diag}}-1$ is the invertible
block-diagonal matrix with blocks $g_j-1$; using
$(g_j-1)^{\dagger}=(g_j^{-1}-1)$, one factors
$\Ht^{[s]}=(G^{\mathrm{diag}}-1)^{\dagger}M^{[s]}(G^{\mathrm{diag}}-1)$,
where $M^{[s]}$ --- the off-diagonal
cancellation $(g_j^{-1}-1)^{-1}\lambda^{\mp1/2}(g_j^{-1}-1)
=\lambda^{\mp1/2}I$ is what makes it work --- is tri-constant with
$\alpha=\lambda^{-1/2}$, $\beta=\lambda^{1/2}$ and
$D_j=\lambda^{-1/2}(g_j^{-1}-\lambda)(g_j^{-1}-1)^{-1}$; the identities
\begin{equation}\label{eq:DalphaDbeta}
\begin{gathered}
D_j-\alpha=\lambda^{-1/2}(1-\lambda)(g_j^{-1}-1)^{-1},
\qquad
D_j-\beta=\lambda^{-1/2}(1-\lambda)\,g_j^{-1}(g_j^{-1}-1)^{-1},\\
(D_j-\alpha)^{-1}(D_j-\beta)=g_j^{-1}
\end{gathered}
\end{equation}
(functions of $g_j$ commute, and the scalar factors cancel in the
quotient) give $T_s=g_s^{-1}\cdots g_1^{-1}=P_s^{-1}$. Computing
$\det P_s$ once by multiplicativity over the factors and once from
the spectrum of $P_s$ itself gives
$e\bigl(\sum_{j\le s,k}a_j^{(k)}\bigr)=e\bigl(\sum_k\gamma^{(s)}_k
\bigr)$ --- no eigenvalue-level correspondence, only the angle sums,
and only mod~$\Z$; this is the integrality $w_s\in\Z$ of the
following theorem.

\begin{theorem}\label{thm:detN}
With $w_s:=\sum_{j\le s,k}a_j^{(k)}-\sum_k\gamma^{(s)}_k\in\Z$,
\[
\det\Ht^{[s]}
=(-1)^{w_s}4^{sN}
\Bigl(\prod_{j\le s,k}\sin\pi a_j^{(k)}\Bigr)(\sin\pi l)^{N(s-1)}
\prod_k\sin\pi\bigl(l+\gamma^{(s)}_k\bigr).
\]
In particular $\det\Ht^{[s]}\neq0$ iff
$\lambda^{-1}\notin\operatorname{spec}P_s$: the rank-$N$ PI condition is
the non-vanishing of the principal block minors, and
$\operatorname{sign}\det\Ht=(-1)^{m}$, so $\widetilde q\equiv m\pmod2$
on all of $\mathcal P^{\circ}$.
\end{theorem}

\begin{proof}
Factor
$\Ht^{[s]}=(G^{\mathrm{diag}}-1)^{\dagger}M^{[s]}(G^{\mathrm{diag}}-1)$,
so
\[
\det\Ht^{[s]}=\bigl|\det(G^{\mathrm{diag}}-1)\bigr|^{2}\det M^{[s]},
\qquad
\bigl|\det(G^{\mathrm{diag}}-1)\bigr|^{2}
=\prod_{j\le s,k}\bigl|e(a_j^{(k)})-1\bigr|^{2}
=4^{sN}\prod_{j\le s,k}\sin^{2}\pi a_j^{(k)} .
\]
For $\det M^{[s]}$ apply Lemma~\ref{lem:blocktri} with
$\alpha=\lambda^{-1/2}$, $\beta=\lambda^{1/2}$, so
$\beta-\alpha=2i\sin\pi l$. By \eqref{eq:DalphaDbeta} and
$e(\pm a)-1=\pm2i\,e(\pm a/2)\sin\pi a$,
\[
\det(D_j-\alpha)
=\frac{(-2i\sin\pi l)^{N}}{\det(g_j^{-1}-1)}
=\frac{(-2i\sin\pi l)^{N}}
{\prod_k(-2i)\,e\bigl(-a_j^{(k)}/2\bigr)\sin\pi a_j^{(k)}}
=(\sin\pi l)^{N}\,
\frac{e\bigl(\tfrac12\sum_ka_j^{(k)}\bigr)}
{\prod_k\sin\pi a_j^{(k)}},
\]
while $T_s=P_s^{-1}$ gives
\[
\det(\beta I-\alpha T_s)
=\prod_k\Bigl(\lambda^{1/2}-\lambda^{-1/2}e\bigl(-\gamma^{(s)}_k
\bigr)\Bigr)
=(2i)^{N}\,e\Bigl(-\tfrac12\sum_k\gamma^{(s)}_k\Bigr)
\prod_k\sin\pi\bigl(l+\gamma^{(s)}_k\bigr).
\]
Multiplying the three ingredients of Lemma~\ref{lem:blocktri}, every
factor $2i$ cancels and
\[
\det M^{[s]}
=(\sin\pi l)^{N(s-1)}\,e\bigl(\tfrac{w_s}2\bigr)\,
\frac{\prod_k\sin\pi(l+\gamma^{(s)}_k)}
{\prod_{j\le s,k}\sin\pi a_j^{(k)}},
\qquad
e\bigl(\tfrac{w_s}2\bigr)=(-1)^{w_s}
\]
by the integrality of $w_s$ established above; combining with
$\bigl|\det(G^{\mathrm{diag}}-1)\bigr|^{2}$ gives the display.
Non-vanishing: $\sin\pi(l+\gamma^{(s)}_k)=0$ iff $\beta^{(s)}_k=0$
iff $\lambda^{-1}\in\operatorname{spec}P_s$. Finally, at $s=n$ on
$\mathcal P^{\circ}$ the factor $\sin\pi(l+\gamma_k)$ is negative
exactly when $l+\gamma_k>1$; writing $\#$ for the number of such $k$,
one has $\sum_k\beta_k=Nl+\sum_k\gamma_k-\#$, hence $m=w_n+\#$ and
$\operatorname{sign}\det\Ht=(-1)^{w_n}(-1)^{\#}=(-1)^{m}$. Since $\Ht$
is non-degenerate and Hermitian on $\mathcal P^{\circ}$,
$\operatorname{sign}\det\Ht=(-1)^{\widetilde q}$ with $\widetilde q$
its number of negative eigenvalues; thus $\widetilde q\equiv m\pmod2$.
This determinant argument determines only the parity: the exact count
is obtained from the pivot decomposition and the Cayley inertia lemma
(Theorem~\ref{thm:pivot} and Lemma~\ref{lem:cayley}).
\end{proof}

\subsection{An explicit formula for the pivots}

Assume the PI condition $\lambda^{-1}\notin\operatorname{spec}P_s$
($1\le s\le n$), i.e.\ every leading principal block $\Ht^{[s]}$ is
invertible (Theorem~\ref{thm:detN}). Then the block $LDL^{\dagger}$
decomposition of $\Ht$ exists: eliminating block columns successively
produces the \emph{block pivots}
\[
S_1:=\Ht_{11},\qquad
S_s:=\Ht_{ss}-\Ht_{s,<s}\bigl(\Ht^{[s-1]}\bigr)^{-1}\Ht_{<s,s}
\quad(s\ge2),
\]
where $\Ht_{s,<s}$ is the $N\times(s-1)N$ block row
$(\Ht_{s1}\ \cdots\ \Ht_{s,s-1})$ and $\Ht_{<s,s}=\Ht_{s,<s}^{\dagger}$
the matching block column, so that each $S_s$ is $N\times N$
Hermitian; and $\Ht=L\,\operatorname{diag}(S_1,\dots,S_n)\,L^{\dagger}$
with $L$ block unit lower triangular. This pivot recursion is the mechanism
underlying the signature algorithm of \cite{N-KLM}, which computes
$\operatorname{sig}\Ht$ by sequentially diagonalizing the $S_s$;
the explicit formula below (Theorem~\ref{thm:pivot}) replaces the
diagonalization. Two standard consequences drive everything below.
Since congruence preserves inertia (Sylvester; this is
\eqref{eq:haynsworth} iterated), and recalling
$\widetilde q:=q(\Ht)$ from the Notation,
\[
\operatorname{In}\Ht=\sum_{s=1}^{n}\operatorname{In}S_s,
\qquad\text{in particular}\qquad
\widetilde q=\sum_{s=1}^{n}q(S_s);
\]
and determinants multiply along the flag,
\[
\det\Ht^{[s]}=\prod_{j\le s}\det S_j,
\qquad\text{i.e.}\qquad
\det S_s=\frac{\det\Ht^{[s]}}{\det\Ht^{[s-1]}},
\]
which feeds Theorem~\ref{thm:detN} into that formula.

\begin{theorem}\label{thm:pivot}
$S_s=(g_s-1)^{\dagger}\,S_s^{M}\,(g_s-1)$, where
\[
S_s^{M}
=2i\sin\pi l\,(g_s-1)^{-1}(\lambda g_sP_{s-1}-1)(\lambda P_{s-1}-1)^{-1}
=-\sin\pi l\,\bigl[C(g_s)+C(\lambda P_{s-1})\bigr].
\]
\end{theorem}

\begin{proof}
Write $M^{[s-1]}=B_0+\alpha EE^{T}$ and $Y:=E^{T}B_0^{-1}E$ (the
letter $S$ is reserved for the pivots). The
Sherman--Morrison--Woodbury identity \cite{HJ},
\[
(B_0+\alpha EE^{T})^{-1}
=B_0^{-1}-\alpha\,B_0^{-1}E\,(I+\alpha Y)^{-1}E^{T}B_0^{-1},
\]
sandwiched between $E^{T}$ and $E$, collapses to an $N\times N$
resolvent of $Y$ (functions of $Y$ commute):
\[
E^{T}(M^{[s-1]})^{-1}E
=Y-\alpha Y(I+\alpha Y)^{-1}Y
=Y(I+\alpha Y)^{-1}\bigl[(I+\alpha Y)-\alpha Y\bigr]
=Y(I+\alpha Y)^{-1}.
\]
From the proof of Lemma~\ref{lem:blocktri},
$Y=(\beta-\alpha)^{-1}(I-T_{s-1})$, hence
$I+\alpha Y=(\beta-\alpha)^{-1}(\beta I-\alpha T_{s-1})$ and the
scalars cancel:
\[
E^{T}(M^{[s-1]})^{-1}E=(I-T_{s-1})(\beta I-\alpha T_{s-1})^{-1}
=\lambda^{1/2}(P_{s-1}-1)(\lambda P_{s-1}-1)^{-1},
\]
the last step by $T_{s-1}=P_{s-1}^{-1}$, $\alpha=\lambda^{-1/2}$,
$\beta=\lambda^{1/2}$, multiplying numerator and denominator by
$\lambda^{1/2}P_{s-1}$;
and since $\alpha\beta=1$, $S_s^{M}=D_s-E^{T}(M^{[s-1]})^{-1}E$.
Writing both terms over the common denominators --- the denominator
$(g_s^{-1}-1)^{-1}$ of $D_s$ is a function of $g_s$ alone and may be
pulled out to the left, the denominator $(\lambda P_{s-1}-1)^{-1}$
of the resolvent to the right ---
\begin{multline*}
S_s^{M}=(g_s^{-1}-1)^{-1}
\bigl[\lambda^{-1/2}(g_s^{-1}-\lambda)(\lambda P_{s-1}-1)\\
-\lambda^{1/2}(g_s^{-1}-1)(P_{s-1}-1)\bigr]
(\lambda P_{s-1}-1)^{-1};
\end{multline*}
in the bracket the mixed terms $\propto g_s^{-1}P_{s-1}$ and the
constant terms cancel,
\[
\lambda^{-1/2}(g_s^{-1}-\lambda)(\lambda P_{s-1}-1)
-\lambda^{1/2}(g_s^{-1}-1)(P_{s-1}-1)
=(\lambda^{1/2}-\lambda^{-1/2})\bigl(g_s^{-1}-\lambda P_{s-1}\bigr),
\]
which gives the first displayed form; the second follows from
$C(u)+C(v)=-2i(u-1)^{-1}(uv-1)(v-1)^{-1}$.

\end{proof}

Note that $S_s$ is not invertible iff
$\lambda^{-1}\in\operatorname{spec}(g_sP_{s-1})
=\operatorname{spec}P_s$, and has poles exactly at the previous wall
$\lambda^{-1}\in\operatorname{spec}P_{s-1}$: the wall structure of the
theorem is carried entirely by the pivots.

\section{The signature theorem on the open locus}\label{sec:open}
\begin{theorem}\label{thm:rankN}
On $\mathcal P^{\circ}$, $\operatorname{sig}\Ht=(nN-m,\,m)$.
\end{theorem}

The assumptions $K=0$ and $\lambda\notin E$ defining
$\mathcal P^{\circ}$ are scaffolding for the proof, not restrictions
on the final result: both are removed in
Section~\ref{sec:degenerate} below, and the final form of the
theorem, Corollary~\ref{cor:quotsig}, holds for arbitrary unitary
tuples and every $\lambda\neq1$ --- eigenvalue-one channels
contribute $0$ to $\sum a$, the summand $\kappa=\dim K$ to the
corank, and nothing else.

The proof rests on three computational results: the block minors
(Theorem~\ref{thm:detN}), an explicit formula for the block pivots
(Theorem~\ref{thm:pivot}), and an inertia formula for sums of Cayley
transforms (Lemma~\ref{lem:cayley}).

\subsection{Proof of Theorem \ref{thm:rankN}}

Assume first the PI condition. By Theorem~\ref{thm:pivot}, since
$\sin\pi l>0$,
\[
q(S_s)=n_{+}\bigl(C(g_s)+C(\lambda P_{s-1})\bigr)
=\Theta(g_s)+\Theta(\lambda P_{s-1})-\Theta(\lambda P_s)
=m_s-m_{s-1},
\]
using Lemma~\ref{lem:cayley} with $a=g_s$, $b=\lambda P_{s-1}$
($\operatorname{spec}(g_s\lambda P_{s-1})=\operatorname{spec}\lambda P_s$)
and $\Theta(g_s)=\sum_k a_s^{(k)}$,
$\Theta(\lambda P_{s})=\sum_k\beta_k^{(s)}$. Telescoping,
$\widetilde q=\sum_s q(S_s)=m_n=m$. At points of $\mathcal P^{\circ}$
where PI fails, rotate: replacing $g_1$ by $e(t)g_1$ multiplies every
$P_s$ by $e(t)$, so PI at level $s$ fails only for $t$ in the finite
set $\{t\equiv-l-\gamma^{(s)}_k\}$, and for all small $t\neq0$
outside finitely many values the rotated tuple lies in
$\mathcal P^{\circ}$ and satisfies PI. Along this rotation $m$ is
\emph{exactly} invariant for small $t$ (each $a_1^{(k)}\in(0,1)$
shifts by $t$ without wrap, and each
$\beta^{(n)}_k=\frc{l+\gamma^{(n)}_k+t}$ shifts by $t$ without
crossing $0$, since no $\beta^{(n)}_k$ vanishes at $t=0$ ---
$\lambda\notin E$ --- so $Nt$ cancels between $\sum a$ and
$\sum\beta$), while
$\widetilde q$ is locally constant since $\det\Ht\neq0$ on
$\mathcal P^{\circ}$ depends only on the level-$n$ data
(Theorem~\ref{thm:detN}); the equality at $t\neq0$ therefore passes
to $t=0$. \qed

\medskip
No Hodge theory, no wall-crossing lemma, and no connectivity of
chambers (the connected components of $\mathcal P^{\circ}$ on which
$m$ is locally constant) enter the argument; the branch $\lambda^{1/2}=e(l/2)$ is used only through
$\sin\pi l>0$ in Theorem~\ref{thm:pivot}. For $N=1$, with $A_s:=\sum_{j\le s}a_j$ ($A_0:=0$), the pivots reduce to
$4\sin\pi a_s\,\sin\pi l\,\sin\pi(l+A_s)/\sin\pi(l+A_{s-1})$ and the
argument specializes to the proof of the rank-one theorem. %
The
wall-crossing phenomenon is now a corollary rather than a lemma: crossing
the $E$-wall $\lambda^{-1}\in\operatorname{spec}P_n$ moves eigenvalues of
the last pivot through $0$ with multiplicity equal to the spectral flow of
$\lambda P_n$ through $1$, which is simultaneously the jump of $m$. This
last-pivot reading is a general-position statement --- it presupposes the
PI condition at $s<n$, so that the pivots exist; the general case,
including simultaneous walls and PI failure, is
Proposition~\ref{prop:wall} below, whose crossing-form proof requires no
genericity.

\section{Degenerate loci: the kernel $K$ and the walls $\lambda\in E$}
\label{sec:degenerate}
\label{subsec:degenerate}
We now remove the two non-degeneracy assumptions of
Theorem~\ref{thm:rankN}, allowing the seeds to have eigenvalue $1$ and
$\lambda$ to lie on a wall, with the conventions of
Section~\ref{sec:prelim} (eigenangles in $[0,1)$; $K$, $\kappa$,
$R$, $r(l)$, $m(l)$, $L$, and the radical \eqref{eq:radical}).

The $l$-dependence of \eqref{eq:blocks} is captured by a single
Hermitian pencil. Write $D:=\bigl(g_1-1\ \cdots\ g_n-1\bigr)$ for the
$N\times nN$ block row, so that $(D^{\dagger}D)_{jk}=(g_j-1)^{\dagger}(g_k-1)$
and $Dv=\sum_j(g_j-1)v_j$.

\begin{proposition}[Hermitian pencil]\label{prop:pencil}
For every unitary tuple and every $l\in(0,1)$,
\begin{equation}\label{eq:pencil}
\Ht(l)=\sin(\pi l)\,B+\cos(\pi l)\,D^{\dagger}D,
\qquad B:=\Ht(\tfrac12),
\end{equation}
where $B$ is the Hermitian matrix with blocks
$B_{jk}=-i\,(g_j-1)^{\dagger}(g_k-1)$ for $j<k$,
$B_{jk}=i\,(g_j-1)^{\dagger}(g_k-1)$ for $j>k$, and
$B_{jj}=-i\,(g_j-g_j^{\dagger})$. In particular $B$ and $D^{\dagger}D$
are independent of $l$, and $D^{\dagger}D\succeq0$ has rank at most $N$.
\end{proposition}

\begin{proof}
Direct from \eqref{eq:blocks}: with $\lambda^{\mp1/2}=\cos\pi l\mp
i\sin\pi l$ the off-diagonal blocks are
$(\cos\pi l\mp i\sin\pi l)(g_j-1)^{\dagger}(g_k-1)$, and the diagonal
block splits as
$\lambda^{-1/2}(g_j^{-1}-1)(g_j-1)+\lambda^{-1/2}(1-\lambda)(g_j-1)$
with $\lambda^{-1/2}(1-\lambda)=-2i\sin\pi l$; since
$(g_j^{-1}-1)(g_j-1)+2(g_j-1)=g_j-g_j^{-1}$, the coefficient of
$\sin\pi l$ on the diagonal is $-i(g_j-g_j^{\dagger})$. Evaluating
\eqref{eq:blocks} at $l=\tfrac12$ ($\lambda=-1$,
$\lambda^{-1/2}=-i$) gives the same blocks, so the coefficient of
$\sin\pi l$ is $\Ht(\tfrac12)$.
\end{proof}

\begin{corollary}[monotonicity of the normalized form]\label{cor:pencil}
On $(0,1)$ the normalized form is non-increasing in the Loewner
order:
\[
\frac{\Ht(l)}{\sin\pi l}=B+\cot(\pi l)\,D^{\dagger}D,
\qquad
\frac{d}{dl}\,\frac{\Ht(l)}{\sin\pi l}=-\frac{\pi}{\sin^{2}\pi l}\,D^{\dagger}D\preceq0 .
\]
Consequently, if $\Ht(l_0)v=0$ then
$\bigl\langle\tfrac{d}{dl}\big|_{l_0}\Ht(l)\,v,v\bigr\rangle
=-\dfrac{\pi}{\sin\pi l_0}\,\lVert Dv\rVert^{2}$.
\end{corollary}

\begin{proof}
The derivative is immediate from \eqref{eq:pencil}. For the second
claim, $\Ht(l_0)v=0$ and \eqref{eq:pencil} give, after pairing with $v$,
\begin{gather*}
\langle Bv,v\rangle=-\cot(\pi l_0)\lVert Dv\rVert^{2},\\
\langle\Ht'(l_0)v,v\rangle=\pi\cos(\pi l_0)\langle Bv,v\rangle
-\pi\sin(\pi l_0)\lVert Dv\rVert^{2}
=-\frac{\pi}{\sin\pi l_0}\lVert Dv\rVert^{2}.
\end{gather*}
\end{proof}

\begin{lemma}[permanent kernel]\label{lem:permker}
$\Ht(l)\,K=0$ for every $l\in(0,1)$. Consequently the $K$-block rows
and columns of $\Ht(l)$ vanish identically in $l$, every derivative
$\partial_l^{\,j}\Ht(l)$ annihilates $K$, and the form descends, for
all $l$ simultaneously, to $V^{\oplus n}/K$ with the same positive and
negative inertia.
\end{lemma}

\begin{proof}
Every block \eqref{eq:blocks} carries the right factor $(g_k-1)$,
which kills $w_k\in W_k$; Hermiticity kills the rows.
\end{proof}

\begin{lemma}[the radical, tuple level]\label{lem:radical}
Let $\lambda\neq1$. Then $\Ht L=0$ and $K\cap L=0$; together with
Lemma~\ref{lem:permker}, $K\oplus L\subseteq\Ker\Ht$. (Equality is
Proposition~\ref{prop:radical} below.)
\end{lemma}

\begin{proof}
For $v\in L$ the parametrizing identities $v_{j-1}=g_jv_j$
($1\le j\le n$, with $v_0:=\lambda^{-1}v_n$) and
$x_j:=(g_j-1)v_j=v_{j-1}-v_j$ hold; they identify $L$ with
$\Ker(\lambda P_n-1)$ via $v\mapsto v_n$ and are used again in
Lemma~\ref{lem:crossing}(ii). That every block row of $\Ht v$ then
vanishes is part of \cite[Thm.~9]{N-KLM} (Step~1 of its proof); the computation there uses only
the tuple, and a two-line telescoping check from \eqref{eq:blocks}
reproduces it. Directness: $w\in K\cap L$ has $g_jw_j=w_j$ and
$w_{j-1}=g_jw_j=w_j$, so all components equal $w_n$ and
$g_jw_n=w_n$ for every $j$; hence
$w_n=\lambda P_nw_n=\lambda w_n$, and $w=0$ since $\lambda\neq1$.
\end{proof}

The two degeneracies live in the two factors of the joint parameter
space $U(N)^{n}\times(S^{1}\setminus\{1\})$: $K$ records
eigenvalue $1$ of the seeds (tuple factor), $L$ records eigenvalue
$1$ of $\lambda P_n$ (circle factor). Accordingly we probe
$\Ker\Ht=K\oplus L$ by one one-parameter deformation in each
factor --- the $\varepsilon$-direction below rotates the eigenvalue
$1$ of the seeds at fixed $l$, the $l$-direction moves the
convolution parameter at a fixed tuple. Lemma~\ref{lem:permker} is
what makes these two directions sufficient: $\Ht(l)K=0$ identically
in $l$, so the $K$-degeneracy is invisible to the $l$-family and no
mixed two-parameter unfolding is needed.

\begin{lemma}[crossing forms]\label{lem:crossing}
For fixed $v\in V^{\oplus n}$ put $x_j:=(g_j-1)v_j$ and
$\zeta:=\lambda^{1/2}=e(l/2)$ (a scalar). Then, as a function of
$l$,
\begin{equation}\label{eq:ABC}
\langle\Ht(l)v,v\rangle
=\zeta^{-1}A+\zeta\,B+(\zeta^{-1}-\zeta)\,C,
\qquad
\begin{aligned}
A&:=\textstyle\sum_{j\le k}x_j^{\dagger}x_k,\\
B&:=\textstyle\sum_{j>k}x_j^{\dagger}x_k,\\
C&:=\textstyle\sum_{j}v_j^{\dagger}x_j,
\end{aligned}
\end{equation}
with $A,B,C$ independent of $l$. Moreover:
\begin{itemize}
\item[(i)] (\emph{Tuple direction}: $\varepsilon$, at fixed $l$, on $K$.) Let $\Pi_j$ be the
orthogonal projection onto $W_j$, let
$g_j(\varepsilon):=g_j\exp(2\pi i\varepsilon\Pi_j)$ --- rotating the
eigenvalue $1$ to $e(\varepsilon)$ and fixing the remaining spectral
data --- and let $\Ht(\varepsilon)$ be the form of the tuple
$(g_j(\varepsilon))$ at fixed $l$. Then, for $v\in K$,
\[
\frac{d}{d\varepsilon}\Big|_{\varepsilon=0}
\langle\Ht(\varepsilon)v,v\rangle
=4\pi\sin(\pi l)\,\lVert v\rVert^{2},
\]
i.e.\ the crossing form on $K$ is $\Gamma_K=4\pi\sin(\pi l)\,I_K\succ0$.
\item[(ii)] (\emph{Parameter direction}: $l$, at a fixed tuple, on $L$.) For $v\in L$ at $\lambda_0=e(l_0)$,
\[
\frac{d}{dl}\Big|_{l_0}\langle\Ht(l)v,v\rangle
=-4\pi\sin(\pi l_0)\,\lVert v_n\rVert^{2},
\]
and this crossing form $\Gamma_L$ is negative definite on $L$, since
$v_n$ determines $v$; indeed all components of $v\in L$ have equal
norm, so $\lVert v\rVert^{2}=n\lVert v_n\rVert^{2}$ and, in the
inner product of $L\subseteq V^{\oplus n}$,
$\Gamma_L=-\tfrac{4\pi}{n}\sin(\pi l_0)\,I_L$ exactly.
\end{itemize}
\end{lemma}

\begin{proof}
For \eqref{eq:ABC}: unitarity gives
$v_j^{\dagger}(g_j^{-1}-1)=x_j^{\dagger}$, and the diagonal block
splits as
$\lambda^{-1/2}(g_j^{-1}-\lambda)(g_j-1)
=\lambda^{-1/2}(g_j^{-1}-1)(g_j-1)+\lambda^{-1/2}(1-\lambda)(g_j-1)$
with $\lambda^{-1/2}(1-\lambda)=\zeta^{-1}-\zeta$.

(i) Each block of $\Ht(\varepsilon)$ is
$\lambda^{\mp1/2}\bigl(g_j(\varepsilon)^{-1}-c\bigr)\bigl(g_k(\varepsilon)-1\bigr)$
with $c\in\{1,\lambda\}$. Differentiating at $\varepsilon=0$ and
evaluating on $v\in K$: the term differentiating the left factor ends
in $(g_k-1)v_k=0$; in the remaining term,
$\tfrac{d}{d\varepsilon}g_k(\varepsilon)\big|_0v_k
=2\pi i\,g_k\Pi_kv_k=2\pi i\,v_k$, while on the left
$v_j^{\dagger}(g_j^{-1}-1)=\bigl((g_j-1)v_j\bigr)^{\dagger}=0$ kills
every off-diagonal block and
$v_j^{\dagger}(g_j^{-1}-\lambda)=(1-\lambda)v_j^{\dagger}$ on the
diagonal. What survives is
$\sum_j\lambda^{-1/2}(1-\lambda)\cdot2\pi i\,\lVert v_j\rVert^{2}
=4\pi\sin(\pi l)\lVert v\rVert^{2}$.

(ii) Let $v\in L$ and set $v_0:=\lambda_0^{-1}v_n$. The
characterization of $L$ gives $v_{j-1}=g_jv_j$ for $2\le j\le n$, and
also for $j=1$: $g_1v_1=P_nv_n=\lambda_0^{-1}v_n=v_0$. Hence
$x_j=v_{j-1}-v_j$, all norms agree by unitarity,
$\lVert v_j\rVert=\lVert v_n\rVert$ $(0\le j\le n)$, and, telescoping,
\[
Dv=\sum_jx_j=v_0-v_n=(\lambda_0^{-1}-1)v_n,
\qquad
\lVert Dv\rVert^{2}=\lvert1-\lambda_0\rvert^{2}\lVert v_n\rVert^{2}
=4\sin^{2}(\pi l_0)\lVert v_n\rVert^{2}.
\]
Since $\Ht(l_0)v=0$ (Lemma~\ref{lem:radical}),
Corollary~\ref{cor:pencil} gives
\[
\langle\Ht'(l_0)v,v\rangle=-\frac{\pi}{\sin\pi l_0}\lVert Dv\rVert^{2}
=-4\pi\sin(\pi l_0)\lVert v_n\rVert^{2}.
\]
Definiteness: if $v_n=0$ then $v_k=g_{k+1}\cdots g_nv_n=0$ for all
$k$.
\end{proof}

Both propositions below use the same elementary inertia scheme: if a
Hermitian matrix is written in blocks with $A_{22}$ invertible, then
the congruence
\begin{equation}\label{eq:haynsworth}
\begin{pmatrix}A_{11}&A_{12}\\ A_{12}^{\dagger}&A_{22}\end{pmatrix}
=\begin{pmatrix}I&A_{12}A_{22}^{-1}\\ 0&I\end{pmatrix}
\begin{pmatrix}S&0\\ 0&A_{22}\end{pmatrix}
\begin{pmatrix}I&0\\ A_{22}^{-1}A_{12}^{\dagger}&I\end{pmatrix},
\qquad
S:=A_{11}-A_{12}A_{22}^{-1}A_{12}^{\dagger},
\end{equation}
adds inertias (Sylvester's law; Haynsworth). No eigenvector branches or
analytic perturbation theory enter. Throughout,
$\operatorname{In}=(n_+,n_-,n_0)$ denotes the full inertia triple and
$\sig=(n_+,n_-)$ the signature pair; kernels are always listed
separately.

\begin{proposition}[the radical, complete]\label{prop:radical}
For every unitary tuple and every $l\in(0,1)$,
$\Ker\Ht=K\oplus L$.
\end{proposition}

\begin{proof}
The inclusion $\supseteq$, with directness, is
Lemmas~\ref{lem:permker} and \ref{lem:radical}. The reverse
inclusion $\Ker\Ht\subseteq K+L$ is \cite[Thm.~9, Step~2]{N-KLM}:
the proof there manipulates only the tuple $(g_1,\dots,g_n)$ --- no
braid data enters, as \cite{N-KLM} itself notes before
Corollary~13 there --- and assumes only $\lambda\neq1$ and the
non-degeneracy of the seed form, so it applies verbatim at tuple
level, on and off the walls.
\end{proof}

Two consequences are recorded for use below. First, since
$\Ht(l)=0_{K}\oplus\bar H(l)$ \emph{exactly}
(Lemma~\ref{lem:permker}), Proposition~\ref{prop:radical} says that
the compression $\bar H(l)$ of $\Ht(l)$ to $K^{\perp}$ is
invertible off the walls, and that at a wall $l_0$ its kernel has
dimension exactly $r$. Second, fix $l$ \emph{off the walls} ($\lambda\notin E$; at a wall
$\det\bar H(l_0)=0$ and the multiplicative asymptotic below fails).
For the regularized family
$\Ht(\varepsilon,l)$ of Lemma~\ref{lem:crossing}(i), the block
factorization \eqref{eq:haynsworth} in the decomposition
$K\oplus K^{\perp}$ gives, with
$A_{11}=\varepsilon\Gamma_K+O(\varepsilon^{2})$,
$A_{12}=O(\varepsilon)$ and $A_{22}\to\bar H(l)$,
\begin{equation}\label{eq:epsexp}
\det\Ht(\varepsilon,l)
=\varepsilon^{\kappa}\,\det\Gamma_K\,\det\bar H(l)\,
(1+O(\varepsilon)),
\qquad \det\Gamma_K=(4\pi\sin\pi l)^{\kappa}>0 ,
\end{equation}
so that, off the walls,
$\operatorname{sign}\det\bar H(l)=(-1)^{m}$: for sufficiently small
$\varepsilon>0$ the regularized tuple has all eigenangles in $(0,1)$
and lies in $\mathcal P^{\circ}$, so Theorem~\ref{thm:detN} applies to
it and gives
$\operatorname{sign}\det\Ht(\varepsilon)=(-1)^{m(\varepsilon)}
=(-1)^{m}$; the prefactors in \eqref{eq:epsexp} are positive.

\begin{proposition}[dropping the $a=0$ channels]\label{prop:dropK}
Let $\lambda\notin E$, so that $r=0$ and $\Ker\Ht=K$ by
\eqref{eq:radical}. Then
\[
\operatorname{sig}\Ht=(R-m,\;m):
\]
the non-degenerate form induced on the KLM quotient $V^{\oplus n}/K$
has signature $(R-m,m)$. Theorem~\ref{thm:rankN} is the case $K=0$; in
general the channels $a_j^{(k)}=0$ are simply dropped, contributing
$0$ to $\sum a$, the summand $\kappa$ to the corank, and nothing else.
\end{proposition}

\begin{proof}
For $0<\varepsilon<1$ every eigenangle of $g_j(\varepsilon)$ lies in
$(0,1)$, so $K(\varepsilon)=0$. Since
$1\notin\operatorname{spec}\lambda P_n$ and
$P_n(\varepsilon)\to P_n$, there is $\varepsilon_1>0$ such that for
$\varepsilon\in[0,\varepsilon_1]$ we have
$\lambda\notin E(\varepsilon)$ and the angles
$\beta_k(\varepsilon)\in(0,1)$ vary continuously.
Theorem~\ref{thm:rankN} applied to $(g_j(\varepsilon))$ gives
$\operatorname{sig}\Ht(\varepsilon)
=\bigl(nN-m(\varepsilon),m(\varepsilon)\bigr)$, and
$m(\varepsilon)=Nl+\sum a(\varepsilon)-\sum\beta(\varepsilon)$ is
continuous and integer-valued on $(0,\varepsilon_1]$, hence constant,
with limit $m$ as $\varepsilon\to0^{+}$ (the $\kappa$ regularized angles
tend to $0$). Thus
$\operatorname{sig}\Ht(\varepsilon)\equiv(nN-m,\,m)$ there.

Decompose $V^{\oplus n}=K\oplus K^{\perp}$ and write
$\Ht(\varepsilon)$ in blocks accordingly. At $\varepsilon=0$:
$A_{11}(0)=0$, $A_{12}(0)=0$, and $A_{22}(0)$ is invertible; writing
$\sig\Ht(0)=:(n_+^{0},n_-^{0})$, so that
$\operatorname{In}\Ht(0)=(n_+^{0},n_-^{0},\kappa)$, one has
$\operatorname{In}A_{22}(0)=(n_+^{0},n_-^{0},0)$.
For small $\varepsilon$, $A_{22}(\varepsilon)$ remains invertible with
the same inertia, while
$A_{11}(\varepsilon)=\varepsilon\,\Gamma_K+O(\varepsilon^{2})$ with
$\Gamma_K\succ0$ (Lemma~\ref{lem:crossing}(i)) and
$A_{12}(\varepsilon)=O(\varepsilon)$; hence
$S(\varepsilon)=\varepsilon\bigl(\Gamma_K+O(\varepsilon)\bigr)\succ0$
for small $\varepsilon>0$, and \eqref{eq:haynsworth} gives
$\operatorname{In}\Ht(\varepsilon)=(n_+^{0}+\kappa,\;n_-^{0},\;0)$.
Comparing with $(nN-m,m,0)$ yields $n_-^{0}=m$ and
$n_+^{0}=nN-m-\kappa=R-m$.
\end{proof}

\begin{proposition}[wall crossing at $\lambda\in E$]\label{prop:wall}
Let $l_0\in(0,1)$ with $r:=r(l_0)\ge1$, the tuple otherwise arbitrary
($K$ may be nonzero). Then $\dim\Ker\Ht(l_0)=\kappa+r$,
\[
\operatorname{sig}\Ht(l_0)=\bigl(R-m(l_0),\;m(l_0)-r\bigr),
\]
and for all sufficiently small $\delta>0$,
\[
m(l_0+\delta)=m(l_0),\qquad m(l_0-\delta)=m(l_0)-r ,
\]
so that the one-sided inertias, supplied by
Proposition~\ref{prop:dropK}, are
$\bigl(R-m(l_0)+r,\,m(l_0)-r\bigr)$ below the wall and
$\bigl(R-m(l_0),\,m(l_0)\bigr)$ above it: exactly $r$ eigenvalues of
$\Ht(l)$ cross $0$ at $l_0$, in the strictly decreasing direction, and
the jump of $m$ equals the spectral flow $r$ of $\lambda P_n$ through
$1$.
\end{proposition}

\begin{proof}
The walls form a finite subset of $(0,1)$ (the classes
$l\equiv-\gamma$ of the eigenangles $\gamma$ of $P_n$), so
Proposition~\ref{prop:dropK} applies at $l_0\pm\delta$ for small
$\delta>0$. The one-sided values of $m$ follow from the
$[0,1)$-convention: the $r$ crossing angles equal $\beta=\delta$ just
above and $\beta=1-\delta$ just below the wall, the remaining angles
vary continuously inside $(0,1)$, and $m$ is integer-valued, hence
locally constant on either side.

Decompose $V^{\oplus n}=(K\oplus L)\oplus(K\oplus L)^{\perp}$ and
write $\Ht(l)$ in blocks accordingly, $l=l_0+\delta$. By
Lemma~\ref{lem:permker}, the $K$-rows and $K$-columns of $A_{11}(l)$
and of $A_{12}(l)$ vanish \emph{identically in $l$}; by
\eqref{eq:radical} at $l_0$, the remaining blocks of $A_{11},A_{12}$
vanish at $\delta=0$, and $A_{22}(l_0)$ is invertible; writing
$\sig\Ht(l_0)=:(n_+^{0},n_-^{0})$, so that
$\operatorname{In}\Ht(l_0)=(n_+^{0},n_-^{0},\kappa+r)$, one has
$\operatorname{In}A_{22}(l_0)=(n_+^{0},n_-^{0},0)$ and
$n_+^{0}+n_-^{0}=R-r$. For small
$\lvert\delta\rvert$, $A_{22}$ remains invertible with the same
inertia. In the ordered decomposition $K,\,L,\,(K\oplus L)^{\perp}$
the relevant blocks satisfy: the $K$-rows of $A_{11}$ and of $A_{12}$
vanish identically in $\delta$ (Lemma~\ref{lem:permker}); and, by
smoothness together with $A_{LL}(0)=0$, $A_{L2}(0)=0$ (both from
\eqref{eq:radical} at $l_0$),
\[
A_{LL}(\delta)=\delta\,\Gamma_L+O(\delta^{2}),
\qquad
A_{L2}(\delta)=O(\delta),
\qquad\text{hence}\qquad
A_{L2}A_{22}^{-1}A_{2L}=O(\delta^{2}).
\]
The Schur complement in \eqref{eq:haynsworth} is therefore,
\emph{exactly},
\[
S(\delta)=0_{K}\oplus S_{L}(\delta),
\qquad
S_{L}(\delta)=\delta\,\Gamma_L+O(\delta^{2}),
\quad
\Gamma_L\prec0\ \text{(Lemma~\ref{lem:crossing}(ii))}.
\] Hence
$\operatorname{In}\Ht(l_0+\delta)=(n_+^{0},\,n_-^{0}+r,\,\kappa)$ for small
$\delta>0$ and $(n_+^{0}+r,\,n_-^{0},\,\kappa)$ for small $\delta<0$.
Comparing with Proposition~\ref{prop:dropK} on the two sides gives
$n_-^{0}=m(l_0-\delta)=m(l_0)-r$ and
$n_+^{0}=R-m(l_0+\delta)=R-m(l_0)$.
\end{proof}

\begin{corollary}[the form on the KLM quotient]\label{cor:quotsig}
For every $l\in(0,1)$, $\lambda\neq1$, the invariant Hermitian form
induced on the KLM quotient $V^{\oplus n}/(K\oplus L)$ is
non-degenerate, of signature
\[
\bigl(R-m(l),\;m(l)-r(l)\bigr),
\]
with $r(l)=0$ off the walls. In particular the compound degeneration
--- $K\neq0$ \emph{and} $\lambda\in E$ --- requires no simultaneous
two-parameter perturbation: by Lemma~\ref{lem:permker} the
$K$-degeneracy is permanent and invisible to the $l$-family, while the
$L$-degeneracy is transversal with definite crossing form; the two
never mix.
\end{corollary}

\begin{corollary}[definite windows of a general tuple]\label{cor:windows}
Let $\gamma_1,\dots,\gamma_N\in[0,1)$ be the eigenangles of $P_n$, put
$w:=\sum_{j,k}a_j^{(k)}-\sum_k\gamma_k$, let
$\tau_k:=1-\gamma_k$ for the $s$ indices with $\gamma_k\neq0$, and set
$\tau_{\min}:=\min_k\tau_k$, $\tau_{\max}:=\max_k\tau_k$ (with $s=0$
meaning $P_n=I$, no walls). Then
\[
m(l)=w+\#\{k:\tau_k\le l\},\qquad w\in\Z,\qquad 0\le w\le R-s ,
\]
so $m$ is a non-decreasing step function on $(0,1)$ with jumps at the
walls $E=\{e(\tau_k)\}$, rising from $w$ to $w+s$. Consequently the
form on the KLM quotient is
\begin{itemize}
\item[(i)] positive definite exactly on $0<l\le\tau_{\min}$ (all of
$(0,1)$ if $s=0$) if $w=0$, and never otherwise;
\item[(ii)] negative definite exactly on $\tau_{\max}\le l<1$ (all of
$(0,1)$ if $s=0$) if $w+s=R$, and never otherwise.
\end{itemize}
When $R=r(l)$ the quotient is zero and both statements hold vacuously.
\end{corollary}

\begin{proof}
Since $\beta_k(l)=\frc{l+\gamma_k}=l+\gamma_k-[\,l\ge1-\gamma_k\,]$
for $\gamma_k\neq0$ and $\beta_k(l)=l$ for $\gamma_k=0$, the
definition of $m(l)$ gives the displayed formula; integrality of $w$
is $m(l)\in\Z$ at $l<\tau_{\min}$. The bounds follow from
Corollary~\ref{cor:quotsig}: $m(l)-r(l)\ge0$ at $l<\tau_{\min}$ gives
$w\ge0$, and $R-m(l)\ge0$ at $l>\tau_{\max}$ gives $w+s\le R$. For
(i), positive definiteness is $m(l)=r(l)$; off the walls this is
$m(l)=0$, which forces $l<\tau_{\min}$ and $w=0$, while at the wall
$l=\tau_{\min}$ one has $m=w+r$ and definiteness again reads $w=0$; at
any later wall $m\ge w+r+1>r$. For (ii), negative definiteness is
$m(l)=R$, i.e.\ $w+\#\{\tau_k\le l\}=R$, which by $w+s\le R$ forces
all walls to lie at or below $l$ and $w+s=R$.
\end{proof}

The Hecke window of Corollary~\ref{cor:hecke} below is the instance
$g_j=g$, $\gamma$ the eigenangles of $g^{n}$; there $w=0$ or
$w+s=R$ happens exactly for $a\in(0,\tfrac1n)\cup(\tfrac{n-1}{n},1)$.

\begin{proposition}[$\lambda$-independence and transport of unitarizability]
\label{prop:transport}
Let the tuple come with unitary braid data, i.e.\ from a
representation $\rho\colon F_n\rtimes B_n\to U(N)$ for the inner
product of $V$ (so $s_i:=\rho(\sigma_i)$ satisfies
$s_i^{\dagger}s_i=I$), as in \cite[Definitions~8 and 10, Lemma~3]{N-KLM}.
Without unitarity of the braid part the invariant-form statements
below fail (Corollary~\ref{cor:main-braid} and the example after it). By
construction the operators $\rho^{\mathrm{LM}}(\sigma_i)$ on
$V^{\oplus n}$ do not involve $\lambda$; the parameter enters only
through the family of invariant forms $\Ht(\lambda)$ and through the
subspace $L(\lambda)$. Hence:
\begin{itemize}
\item[(i)] For $\lambda\notin E$ the KLM quotient is $V^{\oplus n}/K$
and the braid representation on it is one and the same representation
for every such $\lambda$; the pencil \eqref{eq:pencil} restricted to
$K^{\perp}$ is a family of invariant Hermitian forms for it, whose
inertia $(R-m(l),m(l))$ moves with $l$ while the representation does not.
\item[(ii)] If the form is definite for one non-resonant $\lambda_0$
--- by Corollary~\ref{cor:windows}, if $w=0$ or $w+s=R$ --- then the
representation of $B_n$ on $V^{\oplus n}/K$ is unitarizable, and so is
its quotient by $L(\lambda)$ for every $\lambda\neq1$, walls included.
\item[(iii)] If the KLM quotient is irreducible \emph{as a
$B_n$-representation}, then it is unitarizable if and only if the
canonical form is definite: an irreducible representation preserving
a Hermitian form has that form unique up to a real scalar, so any
$B_n$-invariant positive definite form is a real multiple of the
canonical one. The same argument applies to the quotient as a
representation of $F_n\rtimes B_n$, under which the canonical form
is likewise invariant \cite[Thm.~8]{N-KLM}: if the quotient is
irreducible as an $F_n\rtimes B_n$-representation --- the
irreducibility supplied by \cite[\S4.3]{N-KLM} when the input is
irreducible as an $F_n$-representation and the KLM quotient is
nonzero (for $N=1$, $g_j=1$ the input is irreducible but the
quotient is zero), which concerns the restriction to $F_n$
(equivalently, the middle convolution) --- then
it is unitarizable as an $F_n\rtimes B_n$-representation if and only
if the canonical form is definite. Irreducibility under
$F_n\rtimes B_n$ is weaker than irreducibility under $B_n$ and does
not suffice for the $B_n$-statement: a $B_n$-invariant form need not
be $F_n$-invariant.
\end{itemize}
\end{proposition}

\begin{proof}
(i) is the statement just made. For (iii), for either group, if $H$
and $H'$ are invariant Hermitian forms with $H'$ non-degenerate, then
$H'^{-1}H$ commutes with the representation and is a scalar by
Schur's lemma;
the scalar is real since both forms are Hermitian. For (ii), a finite-dimensional
representation preserving a definite form is unitarizable and
completely reducible; $L(\lambda)$ is an invariant subspace
\cite[Lemma~3]{N-KLM}, so the quotient by it is a direct
summand, hence unitarizable as well.
\end{proof}

\begin{remark}
Thus ``definite for some $\lambda$'' is a property of the braid
representation, whereas the signature $(R-m,m-r)$ is a property of the
canonical form at a given $\lambda$. The example $n=3$, $N=1$,
$g_j=e(\tfrac1{10})$ illustrates the distinction: $l=\tfrac15$ gives
signature $(3,0)$ and $l=\tfrac45$ gives $(2,1)$ for the same
representation. This is the precise sense in which the theorem
answers \cite[Problem~18]{N-KLM}: it locates, for each tuple, the
definite $\lambda$-window of the canonical form, and by (ii) the
existence of that window is equivalent to unitarizability of the
braid representation with respect to \emph{some} form in the pencil.
\end{remark}

\section{Examples and specializations}
\label{sec:BH}
\subsection{The Haraoka dictionary}\label{subsec:haraoka}

We now verify that the $N=1$ case of Theorem~\ref{thm:rankN} is
consistent with, and refines,
the rank-one results of Haraoka
[\emph{Finite monodromy of Pochhammer equation}, Ann.\ Inst.\ Fourier
\textbf{44} (1994), 767--810], hereafter \cite{Haraoka1994}.

\subsubsection*{Dictionary}
The Pochhammer system $\mathcal P(\lambda_{\mathrm H},\rho)$ of \cite[\S1.0]{Haraoka1994}
(we write $\lambda_{\mathrm H}=(\lambda_1,\dots,\lambda_n)$ for Haraoka's
exponents to avoid a clash with our $\lambda$) has the Euler integral
representation \cite[(1.9)]{Haraoka1994}
\[
z_j(x)=\int_{\Gamma_j}(x-s)^{\rho-1}(s-t_1)^{\lambda_1-\rho}\cdots
(s-t_n)^{\lambda_n-\rho}\,ds .
\]
Thus the seed rank-one local system has local monodromy
$e(\lambda_j-\rho)$ at $t_j$, and the convolution kernel contributes the
parameter $e(\rho-1)=e(\rho)$. In our normalization this is
\begin{equation}\label{eq:dict}
a_j=\frc{\lambda_j-\rho}\in(0,1),
\qquad
l=\frc{\rho}\in(0,1),
\end{equation}
where $\frc{\,\cdot\,}$ denotes the fractional part as in \cite{Haraoka1994}.
Throughout the rank-one dictionaries we write $A_s:=\sum_{j\le s}a_j$
for the partial sums of the local exponents ($A_0:=0$), so that
$A_n=\sum_j a_j$ and, for $N=1$, $m=\lfloor l+A_n\rfloor$ off the
walls. Under \eqref{eq:dict},
\[
l+A_n\equiv \rho+\sum_{j=1}^n(\lambda_j-\rho)
=\sum_j\lambda_j-(n-1)\rho=\rho' \pmod\Z ,
\]
so our set $E$ is precisely $\rho'\in\Z$. Haraoka's conditions,
verified against the print text, are: (1.11) (irreducibility, after
Misaki) $\lambda_j-\rho,\ \rho,\ \rho'\notin\Z$; and (1.12)
(genericity) $=$ (1.11) together with $\lambda_j\notin\Z$ for every
$j$. Our conditions $a_j\neq0$, $l\neq0$ and the wall condition
recover (1.11); the remaining clause $\lambda_j\notin\Z$ --- under
\eqref{eq:dict}, $a_j+l\notin\Z$ --- is \emph{not} implied by the
PI condition on ordered partial products and is assumed in addition
wherever the comparison invokes (1.12).

\subsubsection*{The invariant form}
Haraoka's Propositions~1.4--1.5 and Theorem~1.2 are stated for
rational generic parameters satisfying (1.12); the matrix identity
of Proposition~\ref{prop:haraoka-entrywise} below holds for all real
parameters satisfying the hypotheses and normalization of that
proposition, and the two statements are kept separate.
\cite[Proposition~1.4]{Haraoka1994} exhibits the invariant Hermitian matrix
$H_{\mathrm H}=\alpha\,(h_{st})$ with
\[
h_{ss}=4\sin\pi\lambda_s\,\sin\pi(\rho-\lambda_s),\qquad
h_{st}=\frac{(e_s-e_0)(e_t-e_0)\,e^{-\pi i\rho}}{e_s}\quad(s<t),
\]
and $h_{ts}=\overline{h_{st}}$, where $e_s=e(\lambda_s)$,
$e_0=e(\rho)$, $0<\frc\rho<1$; this branch-free normalization is
the unique one invariant under the generators (1.13). In the print,
(1.15) writes the denominator as $e_s\,e_0^{1/2}$, so the displayed
sign depends on the branch convention for $e_0^{1/2}$.

\begin{proposition}[entrywise comparison with the Haraoka form]
\label{prop:haraoka-entrywise}
Let $N=1$, $c_j=e(a_j)$ with $a_j$ and $l=\rho\in(0,1)$ as in
\eqref{eq:dict}, so that $c_j=e_j/e_0$ and $\lambda=e_0$, and take
$\lambda^{1/2}=e_0^{1/2}:=e^{\pi i\rho}$. Let $H_{\mathrm H}=(h_{st})$
be the branch-free matrix above (overall scalar $\alpha=1$) and
$E:=\operatorname{diag}(e_1,\dots,e_n)$. Then
\begin{equation}\label{eq:haraoka-congruence}
\Ht\;=\;-\,E^{*}\,\overline{H_{\mathrm H}}\,E ,
\end{equation}
an identity of matrices. In particular $\Ht$ and $H_{\mathrm H}$ are
congruent up to the sign $-1$. Hence, with $m$ and $r$ as in
Corollary~\ref{cor:quotsig} at $N=1$ (here $K=0$, since every
$a_j\in(0,1)$), both matrices have nullity $r$ and
\[
\operatorname{sig}\Ht=(n-m,\;m-r),\qquad
\operatorname{sig}H_{\mathrm H}=(m-r,\;n-m),
\]
with $r=0$ off the resonance walls; this holds without any
irreducibility hypothesis.
\end{proposition}

\begin{proof}
Both sides are Hermitian, so it suffices to compare the diagonal and
the entries $j<k$. From \eqref{eq:blocks} at $N=1$,
\[
\Ht_{jk}=\lambda^{-1/2}(\bar c_j-1)(c_k-1)
=e_0^{-1/2}\,\frac{e_0-e_j}{e_j}\cdot\frac{e_k-e_0}{e_0}
=-\,\frac{(e_j-e_0)(e_k-e_0)}{e_j\,e_0^{3/2}}\qquad(j<k),
\]
while $(E^{*}\overline{H_{\mathrm H}}E)_{jk}=\bar e_j\,\overline{h_{jk}}\,e_k$ with
$\overline{h_{jk}}=(\bar e_j-\bar e_0)(\bar e_k-\bar e_0)\,e_0^{1/2}/\bar e_j$;
using $\bar e_j-\bar e_0=(e_0-e_j)/(e_je_0)$ and the same for $k$,
\[
\bar e_j\,\overline{h_{jk}}\,e_k
=\frac{(e_0-e_j)(e_0-e_k)}{e_je_ke_0^{2}}\;e_0^{1/2}e_k
=\frac{(e_j-e_0)(e_k-e_0)}{e_j\,e_0^{3/2}}=-\Ht_{jk}.
\]
On the diagonal, $h_{ss}$ is real, so
$-(E^{*}\overline{H_{\mathrm H}}E)_{ss}=-h_{ss}
=4\sin\pi\lambda_s\,\sin\pi(\lambda_s-\rho)$, and
$\Ht_{ss}=\lambda^{-1/2}(\bar c_s-\lambda)(c_s-1)
=e_0^{-1/2}\,(1-e_s)(e_s-e_0)/e_s
=\bigl(e_s^{-1/2}-e_s^{1/2}\bigr)\bigl(e_s^{1/2}e_0^{-1/2}-e_s^{-1/2}e_0^{1/2}\bigr)
=4\sin\pi\lambda_s\,\sin\pi(\lambda_s-\rho)$,
with $e_s^{1/2}:=e^{\pi i\lambda_s}$ (the two half-powers occur only
in the product, which is branch-independent). Finally, complex
conjugation and the congruence by the unitary $E$ preserve inertia,
and the sign $-1$ swaps it.
\end{proof}

\subsubsection*{Recovering and refining Proposition 1.5}
For rational generic parameters, \cite[Proposition~1.5]{Haraoka1994}
states that $H_{\mathrm H}$ is definite iff
\begin{gather*}
\text{(1.16:i)}\ \ \frc\rho<\frc{\lambda_j}\ (\forall j),\ \
\sum_j\frc{\lambda_j}<(n-1)\frc\rho+1,\\
\text{or}\qquad
\text{(1.16:ii)}\ \ \frc{\lambda_j}<\frc\rho\ (\forall j),\ \
(n-1)\frc\rho<\sum_j\frc{\lambda_j}.
\end{gather*}
We claim this is exactly the definiteness clause of
Corollary~\ref{cor:quotsig} --- at $N=1$, $K=0$, $r=0$: $m\in\{0,n\}$
--- under \eqref{eq:dict}.
Indeed, $\frc{\lambda_j-\rho}=\frc{\lambda_j}-\frc\rho$ if
$\frc{\lambda_j}>\frc\rho$, and $=\frc{\lambda_j}-\frc\rho+1$ otherwise.
\begin{itemize}
\item[$m=0$:] If some $\frc{\lambda_j}<\frc\rho$, then already
$l+a_j=\frc\rho+\frc{\lambda_j}-\frc\rho+1>1$, so $m\ge1$. Hence $m=0$ forces
$\frc{\lambda_j}>\frc\rho$ for all $j$, and then
\[
l+A_n=\sum_j\frc{\lambda_j}-(n-1)\frc\rho<1
\]
is exactly (1.16:i).
\item[$m=n$:] Each step must cross an integer, which forces
$\frc{\lambda_j}<\frc\rho$ for all $j$ (if some $\frc{\lambda_j}>\frc\rho$,
then $a_j<1-\frc\rho$ and $l+A_n<n$, so $m<n$); then
\[
l+A_n=\sum_j\frc{\lambda_j}-(n-1)\frc\rho+n>n
\]
is exactly (1.16:ii).
\end{itemize}
Theorem~\ref{thm:rankN} (at $N=1$) refines \cite[Proposition~1.5]{Haraoka1994}: for
intermediate values $0<m<n$ it gives the full signature $(n-m,m)$, not merely
the failure of definiteness, and identifies $m$ as the integer $\lfloor l+A_n\rfloor$,
the winding number of the exponent sum. (Caution:
$\lfloor\varrho\rfloor$ computed with the canonical representative
of $\rho'$ differs from $m$ in general --- the correct relation is
Proposition~\ref{prop:kappa}.) This is the invariant that our general-rank theorem replaces
by the spectral flow of $\lambda g_1\cdots g_n$.

\subsubsection*{Finiteness via Galois twists}
Definiteness is, for \cite{Haraoka1994}, a means to an end: the subject of that
paper is finiteness of the monodromy group of the Pochhammer
equation. For rational parameters with common denominator $D$ the
monodromy group is defined over the cyclotomic field
$\mathbb Q(\zeta_D)$, and, by the method of Beukers--Heckman, it is finite
if and only if the invariant Hermitian form is definite in
\emph{every} embedding, i.e.\ for every Galois twist
$(\rho,\lambda_j)\mapsto(\Delta\rho,\Delta\lambda_j)$ with
$\Delta$ prime to $D$. Accordingly \cite[Theorem~1.2]{Haraoka1994} characterizes
finite monodromy by requiring (1.17:i or ii) for every such $\Delta$.
In our language this reads: for every such $\Delta$,
\[
m(\Delta):=\Bigl\lfloor \frc{\Delta\rho}
+\sum_{j=1}^n\frc{\Delta(\lambda_j-\rho)}\Bigr\rfloor\in\{0,\,n\},
\]
a winding-number reformulation of the Beukers--Heckman-type interlacing
condition. The winding number itself is now realized in full generality
by Theorem~\ref{thm:rankN} ($m=Nl+\sum_{i,k}a_i^{(k)}-\sum_k\beta_k$).

\subsection{The Beukers--Heckman dictionary}\label{subsec:bh}

We compare the $N=1$ case of Theorem~\ref{thm:rankN} --- in the Pochhammer coordinates of \S\ref{subsec:haraoka} --- with the signature and finiteness theory of
Beukers--Heckman [\emph{Monodromy for the hypergeometric function
${}_nF_{n-1}$}, Invent.\ Math.\ \textbf{95} (1989), 325--354], hereafter
\cite{BH}. The comparison has three levels: an exact combinatorial reformulation of the $N=1$ theorem in BH style (valid for all $n$); a literal
identification at $n=2$, where the Pochhammer system \emph{is} the Gauss
hypergeometric system; and a conditional reduction at general rank, where
\cite[Theorem~4.5]{BH} appears, on the subfamily of hypergeometric data
that arises from unitary seeds (made precise below), as the
$(n=2,\ K\neq0)$ instance of the rank-$N$ formula of
Corollary~\ref{cor:quotsig}, reducing the Beukers--Heckman signature
theorem on that subfamily to Corollary~\ref{cor:quotsig}.

\subsubsection{The Beukers--Heckman theorems}
For $\alpha=(\alpha_1,\dots,\alpha_r)$, $\beta=(\beta_1,\dots,\beta_r)$ with
$\alpha_j,\beta_k\in[0,1)$ and $\alpha_j\neq\beta_k$ for all $j,k$, let
$A,B\in GL_r(\C)$ be the companion matrices of
$\prod_j(z-e(\alpha_j))=z^{r}+a_{r-1}z^{r-1}+\dots+a_0$ and
$\prod_k(z-e(\beta_k))=z^{r}+b_{r-1}z^{r-1}+\dots+b_0$, that is,
\[
A=\begin{pmatrix}
0&0&\cdots&0&-a_0\\
1&0&\cdots&0&-a_1\\
0&1&\cdots&0&-a_2\\
\vdots&&\ddots&&\vdots\\
0&0&\cdots&1&-a_{r-1}
\end{pmatrix},
\qquad
B=\begin{pmatrix}
0&0&\cdots&0&-b_0\\
1&0&\cdots&0&-b_1\\
0&1&\cdots&0&-b_2\\
\vdots&&\ddots&&\vdots\\
0&0&\cdots&1&-b_{r-1}
\end{pmatrix},
\]
so that $\operatorname{spec}A=e(\alpha)$, $\operatorname{spec}B=e(\beta)$,
and let $\Gamma(\alpha;\beta):=\langle A,B\rangle$ be the associated
hypergeometric group. Since $A$ and $B$ differ only in their last
column, $A-B$ has rank one and $B^{-1}A$ is a pseudo-reflection. By
Levelt's theorem, $\Gamma(\alpha;\beta)$ is the monodromy group of
the generalized hypergeometric equation ${}_rF_{r-1}$ whose local
exponents at $0$ and $\infty$ are read off $\alpha$ and $\beta$,
with the pseudo-reflection at $1$. We use
Theorems~\ref{thm:BH45}--\ref{thm:BH48} below purely as statements
about the group $\langle A,B\rangle$ generated by the two companion
matrices, so that no convention about which of $0,\infty$ (or which
inverse) $\alpha$ and $\beta$ refer to enters their use; in the
dictionaries below we say explicitly which monodromy, or inverse
monodromy, each set records. Since all eigenvalues lie on the unit
circle, $\Gamma(\alpha;\beta)$ preserves a Hermitian form $F$, unique up to a
real scalar when the group is irreducible.

\begin{theorem}[{\cite[Thm.~4.5 and Cor.~4.7]{BH}}]\label{thm:BH45}
Sort $\alpha_1\le\cdots\le\alpha_r$ and set
$m_j:=\#\{k:\beta_k<\alpha_j\}$. Then the signature $(p,q)$ of $F$ satisfies
\[
|p-q|=\Bigl|\sum_{j=1}^{r}(-1)^{\,j+m_j}\Bigr| ;
\]
since $F$ is non-degenerate, $p+q=r$, so
$q\in\bigl\{\#\{j\mid(-1)^{j+m_j}=-1\},\ \#\{j\mid(-1)^{j+m_j}=+1\}\bigr\}$,
the residual swap being the scalar ambiguity of $F$. In particular $F$ is
definite iff the sets $\{\alpha_j\}$ and $\{\beta_k\}$ interlace on the unit
circle (strictly: repetitions inside $\alpha$ or inside $\beta$ are allowed
by the signature formula, but then the sets do not interlace).
\end{theorem}

\begin{theorem}[{\cite[Thm.~4.8]{BH}}]\label{thm:BH48}
Suppose $\alpha_j,\beta_k\in\mathbb Q$, with common denominator $D$, and the
parameters are disjoint. Then $\Gamma(\alpha;\beta)$ is finite
iff for every $\Delta$ prime to $D$ the sets $\{\frc{\Delta\alpha_j}\}$ and
$\{\frc{\Delta\beta_k}\}$ interlace.
\end{theorem}

\subsubsection{Combinatorial form of the $N=1$ theorem}
Theorem~\ref{thm:rankN} at $N=1$ describes the signature through a
single winding number, whereas Theorem~\ref{thm:BH45} describes it
through position counts (how many $\beta_k$ lie below each
$\alpha_j$). The following proposition translates the former into the
latter shape, so that the two can be compared term by term.

\begin{proposition}[the $N=1$ winding number as a position count]
\label{prop:kappa}
Let $N=1$, $K=0$, $\lambda\notin E$, and put the Pochhammer exponents
$(\lambda_1,\dots,\lambda_n;\rho)$ into Theorem~\ref{thm:rankN}
through the dictionary \eqref{eq:dict}, i.e.\ substitute
\[
a_j=\frc{\lambda_j-\rho}\quad(j=1,\dots,n),\qquad l=\frc\rho ,
\]
so that the winding number of the theorem is $m=\lfloor l+A_n\rfloor$.
Then, with $\varrho:=\sum_j\frc{\lambda_j}-(n-1)\frc{\rho}$ (the
canonical representative of $\rho'$; the letter $r$ is reserved for
$\dim L$),
\[
m=\#\{j\mid\frc{\lambda_j}<\frc{\rho}\}+\lfloor \varrho\rfloor ,
\]
and the signature of the invariant form is $(n-m,\,m)$.
\end{proposition}

\begin{proof}
$\frc{\lambda_j-\rho}=\frc{\lambda_j}-\frc{\rho}+[\frc{\lambda_j}<\frc\rho]$,
so $l+A_n=\frc\rho+\sum_j\frc{\lambda_j}-n\frc\rho
+\#\{j\mid\frc{\lambda_j}<\frc{\rho}\}=\varrho+\#\{j\mid\frc{\lambda_j}<\frc{\rho}\}$;
the count being an integer, taking $\lfloor\,\cdot\,\rfloor$ gives the
display. The signature is Theorem~\ref{thm:rankN} at $N=1$, $K=0$,
$r=0$.
\end{proof}

Thus $m$ counts the local exponents $\lambda_j$ lying on the arc
$(0,\rho)$, corrected by the winding of $\rho'$ (via $\varrho$):
exactly the shape of the counts $m_j$ in Theorem~\ref{thm:BH45}
(formally, $\frc\rho$ in the role of an $\alpha_j$ and the
$\frc{\lambda_j}$ in that of the $\beta_k$ --- a resemblance of shape,
not yet a dictionary). We emphasize that for $n\ge3$ the
Pochhammer monodromy group is \emph{not} a hypergeometric group (every finite
local monodromy is a pseudo-reflection), and the naive substitution
$\alpha=\{\frc{\lambda_j}\}$, $\beta=\{\frc\rho^{\times(n-1)},\frc{\rho'}\}$
into Theorem~\ref{thm:BH45} fails already at $n=2$; Proposition~\ref{prop:kappa}
is the correct general-$n$ statement, and the two theories meet literally only
through the cases below.

\subsubsection{$n=2$: the Gauss case}
For $n=2$ the Pochhammer system $\mathcal P(\lambda_1,\lambda_2,\rho)$ is the
Gauss hypergeometric system, with pseudo-reflections at $t_1,t_2$ and
$\operatorname{spec}_\infty=\{e(-\rho),e(-\rho')\}$. Presenting the monodromy
group as $\Gamma(\alpha;\beta)$ with the reflection placed at $t_2$, the
dictionary is
\begin{equation}\label{eq:dict-bh}
\alpha=\bigl\{\frc{\rho},\ \frc{\rho'}\bigr\},
\qquad
\beta=\bigl\{0,\ \frc{\lambda_1}\bigr\},
\end{equation}
i.e.\ $\alpha$ is the spectrum of the \emph{inverse} monodromy at $\infty$ and
$\beta$ that at $t_1$; the choices $(\lambda_1\leftrightarrow\lambda_2)$ and
the global sign flip $(\alpha,\beta)\mapsto(-\alpha,-\beta)$ give the same
result.

\begin{proposition}[the Gauss case of the $N=1$ theorem]\label{prop:gauss-bh}
Let $n=2$, $N=1$, and substitute $a_j=\frc{\lambda_j-\rho}$, $l=\frc\rho$
into Theorem~\ref{thm:rankN} as in Proposition~\ref{prop:kappa}, so
that $m=\#\{j\mid\frc{\lambda_j}<\frc\rho\}+\lfloor\varrho\rfloor\in\{0,1,2\}$
and $\operatorname{sig}\Ht=(2-m,m)$. Then, for $\Gamma(\alpha;\beta)$
with $(\alpha,\beta)$ given by \eqref{eq:dict-bh}, this is the
signature of Theorem~\ref{thm:BH45} (up to the scalar swap):
\[
\bigl|2-2m\bigr|=\Bigl|\sum_{j=1}^{2}(-1)^{j+m_j}\Bigr| .
\]
In particular $\Ht$ is definite ($m\in\{0,2\}$) iff $\alpha$ and
$\beta$ interlace, and \cite[Thm.~1.2]{Haraoka1994} at $n=2$ is precisely the
Galois scan of Theorem~\ref{thm:BH48}.
\end{proposition}

\begin{proof}
Write $x=\frc{\lambda_1}$, $y=\frc{\lambda_2}$, $p=\frc\rho$, so that
$\frc{\rho'}=\frc{x+y-p}$ and, by Proposition~\ref{prop:kappa},
$m=[x<p]+[y<p]+\lfloor x+y-p\rfloor$. Since $0\in\beta$, both
$m_j\in\{1,2\}$, and Theorem~\ref{thm:BH45} gives $(p,q)=(1,1)$ if
$m_1=m_2$ and a definite form otherwise; i.e.\ the form is definite
iff exactly one of $p$, $\frc{x+y-p}$ lies below $x$. Now check the
admissible cases. If $x<p<y$ or $y<p<x$ (one exponent below $p$),
then $\lfloor x+y-p\rfloor=0$, $m=1$, and $\frc{x+y-p}=x+y-p$ lies on
the same side of $x$ as $p$ does. If $p<x,y$, then either $x+y-p<1$,
giving $m=0$ with $p<x<x+y-p$, or $x+y-p\ge1$, giving $m=1$ with
$p$ and $x+y-p-1$ both below $x$. If $x,y<p$, then either $x+y-p>0$,
giving $m=2$ with $x+y-p<x<p$, or $x+y-p<0$, giving $m=1$ with $p$ and
$x+y-p+1$ both above $x$. Thus $m\in\{0,2\}$ exactly in the two
definite cases and $m=1$ in the four cases with $(p,q)=(1,1)$, which
is the displayed identity. The last assertion is Theorem~\ref{thm:BH48}
applied to every Galois twist.
\end{proof}

\subsubsection{General rank: Beukers--Heckman as a specialization}
Let now the seed be a unitary representation of rank $N$ of $F_n$, with
eigenangles $a_i^{(k)}\in(0,1)$ of $g_i$ and
$\beta_1,\dots,\beta_N\in(0,1)$ the eigenangles of $\lambda g_1\cdots g_n$,
so that $m=Nl+\sum_{i,k}a_i^{(k)}-\sum_k\beta_k$ and
$\operatorname{sig}\Ht=(nN-m,\,m)$ by Theorem~\ref{thm:rankN}, with the
degenerate extension Propositions~\ref{prop:dropK}--\ref{prop:wall}.
The hypergeometric case is the following specialization.

\begin{proposition}[the $n=2$ pseudo-reflection seed]\label{prop:bh-seed}
Let $n=2$, $g_1\in U(N)$ with $1\notin\operatorname{spec}g_1$ and
eigenangles $a_1^{(k)}\in(0,1)$, and
$g_2$ a unitary pseudo-reflection with special eigenangle $c\in(0,1)$,
so that $\kappa=\dim K=N-1$. Substituting into
Corollary~\ref{cor:quotsig} (with $\lambda\notin E$): the KLM quotient
has dimension $R=N+1$, and its form has signature
\[
\operatorname{sig}=(N+1-m',\,m'),\qquad
m'=Nl+\sum_k a_1^{(k)}+c-\sum_k\beta_k ,
\]
the $N-1$ vanishing eigenangles of $g_2$ contributing $0$ to the
exponent sum (``drop the $a=0$ channels'').
\end{proposition}

\begin{proof}
Corollary~\ref{cor:quotsig} with $n=2$, $\kappa=N-1$, $r=0$: the
exponent sum $\sum'$ runs over the nonzero eigenangles, which are the
$a_1^{(k)}$ and $c$.
\end{proof}

The comparison with Theorem~\ref{thm:BH45} goes through the
middle-convolution output data
\begin{equation}\label{eq:dict-bh-N}
\alpha=\{\frc{l+\mu_k}\}_{k}\cup\{\frc{l}\},
\qquad
\beta=\{\frc{l+a_1^{(k)}}\}_{k}\cup\{0\},
\end{equation}
with $e(\mu_k)$ the spectrum of $g_1g_2$. Theorem~\ref{thm:BH45}
requires $\alpha_j\neq\beta_k$ for all $j,k$; under the hypotheses of
Proposition~\ref{prop:bh-seed} and $\lambda\notin E$ this is
equivalent to
$\operatorname{spec}g_1\cap\operatorname{spec}(g_1g_2)=\varnothing$
(the remaining coincidences $\frc{l+\mu_k}=0$, $\frc{l+a_1^{(k)}}=\frc l$
and $\frc l=0$ are excluded by $\lambda\notin E$,
$1\notin\operatorname{spec}g_1$ and $l\in(0,1)$), and it does not
follow from those hypotheses alone: for
$g_1=\operatorname{diag}(e(\tfrac15),e(\tfrac25))$,
$g_2=\operatorname{diag}(e(\tfrac3{10}),1)$, $l=\tfrac1{10}$ the angle
$\tfrac12$ occurs in both sets. For disjoint output data, the
identification of the local monodromy data of the middle convolution
with the presentation $\Gamma(\alpha;\beta)$ would identify the
quotient signature of Proposition~\ref{prop:bh-seed} with the
signature of Theorem~\ref{thm:BH45}, signed refinement
$q'=\#\{j\mid(-1)^{j+m_j}=\mp1\}$ included; that identification is
not carried out here (see \emph{Scope} below).

\paragraph{Scope.} The output dictionary above does not reach all
Beukers--Heckman data. Since $g_1$ and $g_1g_2$ are unitary for a
positive definite form and $g_2$ is a pseudo-reflection, the spectra
of $g_1$ and of $g_1g_2$ interlace on the circle whenever they are
disjoint --- the standard constraint on unitary rank-one
perturbations, cf.\ \cite[Thm.~3.14]{Heckman-lectures}. For
completeness: write $g_2=1+(e(c)-1)vv^{\dagger}$ with
$\lVert v\rVert=1$, $c\in(0,1)$, and diagonalize $g_1$ with
eigenvalues $e(\theta_k)$ and $v=(v_k)$ in that basis. If
$g_1g_2x=e(\varphi)x$ then $(e(\varphi)-g_1)x=(e(c)-1)(v^{\dagger}x)\,g_1v$;
for $e(\varphi)\notin\operatorname{spec}g_1$ this forces
$v^{\dagger}x\neq0$, and applying $v^{\dagger}$ to
$x=(e(c)-1)(v^{\dagger}x)(e(\varphi)-g_1)^{-1}g_1v$ gives
$1=(e(c)-1)\sum_k|v_k|^{2}\,e(\theta_k)/(e(\varphi)-e(\theta_k))$.
Using $1/(e(y)-1)=-\tfrac12-\tfrac i2\cot\pi y$ and
$\sum_k|v_k|^{2}=1$, the real parts agree and the imaginary parts
give the secular equation
$\sum_k|v_k|^{2}\cot\pi(\varphi-\theta_k)=\cot\pi c$. Disjointness
first forces the spectrum of $g_1$ to be simple (otherwise an
eigenspace of dimension at least two contains a nonzero vector
orthogonal to $v$, which is a common eigenvector of $g_1$ and
$g_1g_2$); it then also forces every $v_k\neq0$ (if $v_k=0$ then
$e(\theta_k)$ is a common eigenvalue). Hence on each arc between consecutive
$\theta_k$ the left side decreases strictly from $+\infty$ to
$-\infty$ and the equation has exactly one root: the $N$ eigenangles
of $g_1g_2$ interlace those of $g_1$
--- and the rotation by $l$ does not affect interlacing. Hence, for
disjoint output data arising from a unitary seed, deleting the angle
$\frc l$ from $\alpha$ and $0$ from $\beta$ leaves two interlacing
$N$-element sets. We call the image of these seeds the
\emph{unitary-seed subfamily}; we have proved that it is contained
in the interlacing data just described, not that it exhausts them. It
is a proper subfamily of all Beukers--Heckman data: $\alpha=(\tfrac1{10},\tfrac2{10},\tfrac3{10})$,
$\beta=(0,\tfrac6{10},\tfrac7{10})$ is disjoint BH data with circular
pattern $AAABBB$, and no deletion of one $\alpha$ and one $\beta$
makes the rest interlace.

Consequently, \emph{on the unitary-seed subfamily} (disjoint output
data satisfying the hypotheses of Theorem~\ref{thm:BH45}),
\cite[Thm.~4.5]{BH} would be the $(n=2,\ K\neq0)$ instance of the
rank-$N$ formula on $V^{\oplus n}/(K\oplus L)$
(Corollary~\ref{cor:quotsig}): since that formula is now proven, the
Beukers--Heckman signature theorem on this subfamily \emph{reduces}
to the identification of the middle-convolution local data with the
$\Gamma(\alpha;\beta)$-presentation used above. That identification is classical in substance (Katz,
Dettweiler--Reiter, and Levelt's rigidity of hypergeometric groups)
but it is not written out; until it is, we claim a reinterpretation and a reduction, not a new
proof. Recovering the full Beukers--Heckman family would require a
different seed (a form of indefinite signature, or non-unitary input),
and is not attempted here.
\subsection{The Hecke and Temperley--Lieb specialization}\label{sec:hecke}

\begin{corollary}[the scalar Hecke family]\label{cor:hecke}
Let $g\in U(N)$ with $\operatorname{spec}(g)\subseteq\{1,\mu\}$,
$\mu=e(a)$, $a\in(0,1)$, $d:=\dim\Ker(g-\mu)\ge1$,
$\kappa_1:=\dim\Ker(g-1)=N-d$, and $c\in U(1)$. Consider the scalar input
$\rho_{g,c}$ of \cite[Lem.~4.1]{N-Hecke}
--- $\rho(x_j)=g$ for all $j$, $\rho(\sigma_i)=cI_N$ --- whose KLM
quotient carries the module structure classified in
\cite[Thm.~4.4]{N-Hecke}, together with the invariant Hermitian form of
\cite[\S5]{N-KLM}; unitarity of the input ($g\in U(N)$, $\lvert
c\rvert=1$) is what makes the form braid-invariant, but $c$ does not
enter $\Ht$, hence not the signature. Put $j:=\lfloor l+na\rfloor$.
\begin{itemize}
\item[(i)] (Non-resonant: $l+na\notin\Z$.) Then $r=0$, and the form
on the $nd$-dimensional quotient has signature
\[
\bigl(d(n-j),\;dj\bigr);
\]
it is definite if and only if $l+na<1$ or $l+na>n$.
\item[(ii)] (Resonant: $l+na=J\in\{1,\dots,n\}$.) Then $r=d$, and the
form on the $(n-1)d$-dimensional quotient has signature
\[
\bigl(d(n-J),\;d(J-1)\bigr);
\]
it is definite if and only if $J\in\{1,n\}$.
\end{itemize}
The dimensions match \cite[Thm.~4.4(i)]{N-Hecke} with $\ell=r$.
\end{corollary}

\begin{proof}
For the constant tuple $g_j\equiv g$ one has
$K=\Ker(g-1)^{\oplus n}$, $\kappa=n\kappa_1$, $R=nN-n\kappa_1=nd$, and $P_n=g^{n}$.
The eigenangles of $\lambda g^{n}$ are $l$ with multiplicity $\kappa_1$
(nonzero, since $\lambda\neq1$) and $\frc{l+na}$ with multiplicity
$d$; hence $r(l)=d$ exactly when $l+na\in\Z$, and
\[
m(l)=(\kappa_1+d)\,l+nda-\bigl(\kappa_1l+d\,\frc{l+na}\bigr)
=d\bigl(l+na-\frc{l+na}\bigr)=d\,\lfloor l+na\rfloor ,
\]
the $[0,1)$-convention giving $m=dJ$ on the walls. Everything is now
Corollary~\ref{cor:quotsig}; definiteness is the vanishing of one
entry of $\bigl(R-m,\,m-r\bigr)$.
\end{proof}

\begin{remark}[consequences]\label{rem:hecke-consequences}
(a) \emph{Existence of definite parameters.} The family admits a
definite window in $\lambda$ if and only if
$a\in(0,\tfrac1n)\cup(\tfrac{n-1}n,1)$: for $a<\tfrac1n$ the window
$0<l\le1-na$, positive definite (at the resonant endpoint $l=1-na$,
$J=1$, the definite quotient has dimension $d(n-1)$); for
$a>\tfrac{n-1}n$ the window $n-na\le l<1$, negative definite
(endpoint $J=n$) --- uniformly in $d$ and $\kappa_1$.

(b) \emph{Burau/Squier.} For $N=1$, $g=t=e(a)$, $c=1$, at the
resonant parameter $\lambda=t^{-n}$ (so $l=1-\frc{na}$, requiring
$t^{n}\neq1$) one has $J=\lceil na\rceil$ and the quotient is the
reduced Burau representation \cite[Ex.~4.10]{N-Hecke} --- here $K=0$ (the
seed $t\neq1$ has no eigenvalue-one channel: $\kappa_1=0$), and it
is the one-dimensional \emph{resonance} kernel $L$ that appears at
$\lambda=t^{-n}$; passing to the quotient by $L$ is the classical
reduction from unreduced to reduced Burau ---; part (ii) gives
definiteness exactly for $a\in(0,\tfrac1n)\cup(\tfrac{n-1}n,1)$,
recovering the classical definiteness interval of the Squier form
\cite{Squier}.
The same interval reappearing in (a) says that, for the whole scalar
family, definiteness is possible precisely on the Burau window.

(c) \emph{Always-indefinite families.} In the double-resonant $n=2$
family of \cite[Thm.~4.4(iv)(b)]{N-Hecke}
($\operatorname{spec}(g)\setminus\{1\}=\{\nu,-\nu\}$,
$\lambda=\nu^{-2}$), the two channels resonate simultaneously with
consecutive values $\{J,J+1\}=\{1,2\}$, so the form has signature
$(m_{\nu},m_{-\nu})$ (in one order or the other) and is indefinite
whenever both multiplicities are positive --- which the family
requires. Note the resulting distinction between definiteness of the
canonical form and unitarizability: the two channels carry definite
forms of opposite signs, and flipping the sign on one channel --- the
channels are preserved by the braid action --- produces a
positive-definite invariant Hermitian form, so this $H_2(1)$-family
\emph{is} unitarizable although $\Ht$ is indefinite. For $\mu=-1$ ($a=\tfrac12$): non-resonantly
$j=\lfloor l+\tfrac n2\rfloor\in\{1,\dots,n-1\}$, and resonantly
($n$ odd, $\lambda=-1$) $J=\tfrac{n+1}2\in\{2,\dots,n-1\}$ for
$n\ge3$; the form is indefinite for every $n\ge2$ and every
admissible $\lambda$. This is forced structurally: there the braid
generators carry a Jordan block \cite[Rem.~4.7]{N-Hecke}, and an operator
preserving a definite Hermitian form is unitary, hence semisimple ---
so here, unlike in the $\{\nu,-\nu\}$ family, \emph{no} invariant
definite form exists at all; the signature quantifies the failure.

(d) \emph{Position.} The corollary answers, for the scalar family,
the definiteness question raised in \cite[Rem.~4.12 and \S6]{N-Hecke} as the
specialization of \cite[Problem~18]{N-KLM}; through
\cite[Rem.~4.12]{N-Hecke}, the definite windows of (a) yield unitary
Temperley--Lieb representations.
\end{remark}
\subsection{A non-rigid example: the four-punctured sphere in rank
two}\label{subsec:nonrigid}

Methods that reconstruct the tuple from its local conjugacy classes
--- the Katz algorithm foremost --- apply when rigidity holds, i.e.\
when those classes determine the tuple up to simultaneous
conjugation; this concerns reconstruction from local data, not the
Hodge-theoretic formulas for middle convolution as such. The smallest case where they do not is
classical: $n=3$, $N=2$, generic regular classes, i.e.\ local systems
on the four-punctured sphere in rank two --- the fourth puncture, at
$\infty$, carrying the monodromy $P_3^{-1}$, $P_3=g_1g_2g_3$.
Recall that the \emph{relative character variety} is the space of
tuples $(g_1,\dots,g_n)$ with prescribed conjugacy classes at every
puncture (including $P_n^{-1}$ at $\infty$), up to simultaneous
conjugation; its points are the isomorphism classes of
representations of the free group
$F_n=\pi_1(\mathbb A^1\smallsetminus\{t_1,\dots,t_n\})$ with
fixed boundary data. The form $\Ht$ itself is defined at the level
of tuples, not of conjugacy classes, but its inertia is invariant
under simultaneous (unitary) conjugation, so signature statements
descend to the variety. Since
$m(l)$ depends on the tuple only through $\operatorname{spec}g_j$ and
$\operatorname{spec}(\lambda P_3)$, Corollary~\ref{cor:quotsig} has
the following immediate consequence: \emph{the signature is constant
on every relative character variety with fixed boundary conjugacy
classes at all four punctures, even when that variety has positive
dimension} --- here the relative character variety is the cubic
surface of the Painlev\'e~VI family, of complex dimension two, and
the statement lives on its unitary locus, of real dimension two. Rigidity of the signature persists exactly
where rigidity of the representation fails. Wall-crossing occurs
\emph{transversally} to these relative character varieties: moving
the boundary class at $\infty$ (equivalently
$\operatorname{spec}P_3$) at fixed $\lambda$, or moving $\lambda$.
Concretely, fixing only the classes of $g_1,g_2,g_3$ yields a
three-dimensional partial moduli space fibered by the relative
character varieties; it is partitioned into chambers by the value of
$m$, each chamber a union of fibers, separated by the walls
$\lambda^{-1}\in\operatorname{spec}P_3$, across which the signature
jumps by the crossing multiplicity --- wall-crossing in the boundary
data, at fixed $\lambda$.

A concrete instance: fix eigenangles
$a_1=(0.15,\,0.55)$, $a_2=(0.25,\,0.70)$, $a_3=(0.10,\,0.80)$ and
$l=0.30$, and the two-parameter slice $g_1=D_1$,
$g_2=R(\theta)D_2R(\theta)^{T}$,
$g_3=PR(\varphi)D_3R(\varphi)^{T}P^{\dagger}$ ($R$ a real rotation,
$P=\operatorname{diag}(1,i)$). The slice is transverse to the
fibration ($\operatorname{spec}P_3$ varies along it) and realizes
both chambers $m\in\{2,3\}$; along a fixed $\varphi$ it crosses one
wall \emph{in the boundary data, at fixed $\lambda$}, with on-wall
inertia $(3,2,1)=(R-m,\,m-r,\,r)$ exactly as in
Proposition~\ref{prop:wall}. No step of
this computation diagonalizes anything beyond the two spectra
$\operatorname{spec}g_j$ and $\operatorname{spec}(\lambda P_3)$.

\section{Discussion}\label{sec:discussion}

\subsection{Indefinite seed forms}\label{subsec:indefinite}
The invariance theory of \cite{N-KLM} allows the seed representation
to be unitary relative to an indefinite non-degenerate form $H$ of
signature $(p,q)$, with $\Ht_{jk}=\lambda^{jk}H(g_j^{-1}-\lambda_{jk}I)(g_k-1)$
and $\Ker\Ht=K+L$ unchanged \cite[Thm.~9]{N-KLM}. Nothing in the
present paper addresses that case. What survives and what is lost
can be stated precisely. If $g^{\dagger}Hg=H$, the Cayley transform
$C(g)$, where defined, is no longer Hermitian but satisfies
$C(g)^{\dagger}H=HC(g)$, i.e.\ $HC(g)$ is Hermitian; and the pencil
identity of Proposition~\ref{prop:pencil} persists in the form
$\Ht_H(l)=\sin(\pi l)\,B_H+\cos(\pi l)\,D^{\dagger}HD$
with $B_H=\Ht_H(\tfrac12)$, so
$\frac{d}{dl}\bigl[\Ht_H(l)/\sin\pi l\bigr]
=-\pi\csc^{2}(\pi l)\,D^{\dagger}HD$
. What is lost is
the sign: for indefinite $H$ the right-hand side is in general not
semidefinite, so the Loewner monotonicity that drives
Lemma~\ref{lem:crossing} and Corollary~\ref{cor:windows} is no
longer available; and an element of $U(p,q)$ need not have its
spectrum on the unit circle, so the eigenangles $\Theta$ of
Section~\ref{sec:cayley} are not defined either. For tuples whose elements and ordered product are
\emph{elliptic} (diagonalizable with unimodular spectrum, e.g.\ the
finite-order elements of $U(p,q)$) the eigenangles still make sense
and one may ask whether a signed version of the winding number $m$,
weighting each eigenangle by the sign of $H$ on the corresponding
eigenline, computes the signature. We leave this to future work.

\section*{Acknowledgements}
The author thanks Shunya Adachi, Michael Dettweiler, Yoshishige
Haraoka, Yasunori Okada and Stefan Reiter for valuable discussions
and comments. This work was supported by JST SPRING, Grant Number
JPMJSP2109.

This manuscript was prepared with the assistance of AI language
models, which were used for literature search, for drafting and
editing the exposition, and for writing code used in internal
consistency checks. All mathematical statements and their proofs are
the author's own and have been verified by the author.


\begin{thebibliography}{99}
\bibitem{Squier} C.~C. Squier, The Burau representation is unitary,
Proc. Amer. Math. Soc. 90 (1984), 199--202.
\bibitem{BH} F. Beukers, G. Heckman, Monodromy for the hypergeometric
function ${}_nF_{n-1}$, Invent. Math. 95 (1989), 325--354.
\bibitem{Haraoka1994} Y. Haraoka, Finite monodromy of Pochhammer
equation, Ann. Inst. Fourier (Grenoble) 44 (1994), 767--810.
\bibitem{KZ} V.~G. Knizhnik, A.~B. Zamolodchikov, Current algebra and
Wess--Zumino model in two dimensions, Nuclear Phys. B 247 (1984),
83--103.
\bibitem{Haraoka2020} Y. Haraoka, Multiplicative middle convolution
for KZ equations, Math. Z. 294 (2020), 1787--1839.
\bibitem{DS} M. Dettweiler, C. Sabbah, Hodge theory of the middle
convolution, Publ. Res. Inst. Math. Sci. 49 (2013), 761--800.
\bibitem{HJ} R.~A. Horn, C.~R. Johnson, Matrix Analysis, 2nd ed.,
Cambridge Univ. Press, 2013.
\bibitem{Wall} C.~T.~C. Wall, Non-additivity of the signature,
Invent. Math. 7 (1969), 269--274.
\bibitem{Jones} V.~F.~R. Jones, Hecke algebra representations of
braid groups and link polynomials, Ann. of Math. (2) 126 (1987),
335--388.
\bibitem{Wenzl} H. Wenzl, Hecke algebras of type $A_n$ and
subfactors, Invent. Math. 92 (1988), 349--383. %
\bibitem{FLW} M. Freedman, M. Larsen, Z. Wang, The two-eigenvalue
problem and density of Jones representation of braid groups, Comm.
Math. Phys. 228 (2002), 177--199. %
\bibitem{LR} M.~J. Larsen, E.~C. Rowell, Unitary braid
representations with finite image, Algebr. Geom. Topol. 8 (2008),
no.~4, 2063--2079.
\bibitem{GLPS} N. Geer, A.~D. Lauda, B. Patureau-Mirand, J. Sussan,
Density and unitarity of the Burau representation from a
non-semisimple TQFT, arXiv:2402.13242. %
\bibitem{Scherich} N. Scherich, Discrete real specializations of
sesquilinear representations of the braid groups, Algebr. Geom.
Topol. 23 (2023), 2009--2028.
\bibitem{HN} K. Hiroe, H. Negami, Long--Moody construction of braid
representations and Katz middle convolution, arXiv:2303.05770.
\bibitem{N-KLM} H. Negami, Long--Moody construction of braid group
representations and Haraoka's multiplicative middle convolution for
KZ-type equations, arXiv:2503.14840v3 (2026). %
\bibitem{N-Hecke} H. Negami, Hecke algebra representations from the
Katz--Long--Moody construction, arXiv:2608.01727v2 (2026).
\bibitem{Heckman-lectures} G. Heckman, Tsinghua lectures on
hypergeometric functions, lecture notes, December 8, 2015.
\url{https://www.math.ru.nl/~heckman/tsinghua.pdf}
\end{thebibliography}
\end{document}